\documentclass[letterpaper]{article} 
\usepackage{aaai2027}  
\usepackage[hyphens]{url}  
\usepackage{graphicx} 
\usepackage{natbib}  
\usepackage{caption} 
\usepackage{tikz}
\usetikzlibrary{positioning,arrows.meta, fit}
\usepackage{amsmath}
\usepackage{amssymb}

\usepackage{booktabs}
\usepackage{tabularx}
\usepackage{multirow}
\usepackage{siunitx}

\usepackage{amsthm,amssymb}
\theoremstyle{definition}
\newtheorem{definition}{Definition}
\newtheorem{hypothesis}{Hypothesis}

\usepackage{algorithm}
\usepackage{algpseudocode}

\usepackage{newfloat}
\usepackage{listings}
\DeclareCaptionStyle{ruled}{labelfont=normalfont,labelsep=colon,strut=off} 
\floatstyle{ruled}
\newfloat{listing}{tb}{lst}{}
\floatname{listing}{Listing}

\usepackage{bm}

\newtheorem{proposition}{Proposition}

\usepackage{acronym}
\acrodef{DFM}{Discriminative Flow
Matching}
\acrodef{CFM}{Conditional Flow Matching}
\acrodef{FM}{Flow Matching}
\acrodef{SE}{speech enhancement}
\acrodef{CNF}{Continuous Normalizing Flow}
\acrodef{RF}{Rectified Flows}
\acrodef{NFE}{number of function evaluations}
\acrodef{SNR}{signal-to-noise ratio}
\acrodef{OT-CFM}{Optimal Transport-CFM}  
\acrodef{DFSR}{Discriminative Flow-State Representations}

\newcommand{\Wtwo}{\mathcal{W}_2}
\newcommand{\flowstate}{\mathcal{S}}
\newcommand{\pdata}{p_1}

\title{Discriminative Flow Matching: Beyond Time-Conditioning in Generative Restoration via Flow-State Representations}
\author{Shrishti Saha Shetu$^{1}$,
Emanu\"{e}l A. P. Habets$^{1,2}$, Andreas Brendel$^{1}$
}
\affiliations{$^{1}$International Audio Laboratories Erlangen\thanks{A joint institution of Friedrich-Alexander-Universit\"{a}t Erlangen-N\"{u}rnberg
(FAU) and
Fraunhofer
IIS.}, Am Wolfsmantel 33, 91058 Erlangen, Germany \\
    $^{2}$Fraunhofer IIS, Am Wolfsmantel 33, 91058 Erlangen, Germany

}

\begin{document}

\maketitle

\begin{abstract}

Existing Conditional Flow Matching (CFM) formulations describe transport progress using an explicit interpolation coordinate, commonly interpreted as a time coordinate, implicitly assuming that a single global variable adequately describes a sample's position along the generative trajectory. However, in restoration tasks, transport progress is inherently sample-dependent, as the initial distribution often exhibits varying statistical dependencies with the target distribution. Consequently, samples at the same interpolation coordinate may differ substantially in degradation level, distance of their marginal distribution to the target distribution, and restoration difficulty. In this work, we investigate whether signal representations learned by discriminatively trained models provide a meaningful description of the state of generative transport within the context of CFM-based restoration tasks. Through a systematic latent-space analysis of discriminative restoration models, we demonstrate that discriminative representations organize according to degradation severity and follow a consistent trajectory toward the clean-data manifold during the generative process. Motivated by these observations, we introduce the Discriminative Flow-State Hypothesis, which suggests that discriminative representations encode a transport state that governs generative restoration. Based on this hypothesis, we propose Discriminative Flow Matching, a framework that conditions the Flow-Matching velocity field on Discriminative Flow-State Representations instead of explicit time coordinates. Experiments on speech enhancement and image denoising demonstrate that these representations sufficiently characterize restoration progress, support adaptive inference, and consistently outperform various CFM and diffusion-related baselines. Our findings suggest that discriminative representations provide an effective state-aware alternative to explicit time conditioning and offer a novel perspective on the relationship between discriminative and CFM-based generative modeling.

\end{abstract}


\section{Introduction}

The task of many deep generative models can be viewed as the problem of transporting probability mass from a source distribution toward a target distribution. Recent advances in diffusion models, score-based generative models, and Flow-based methods have achieved remarkable success across image synthesis, audio generation, restoration, and inverse problems~\cite{ho2020ddpm,song2021score,lipman2023flow}. Among these approaches, \ac{CFM}~\cite{lipman2023flow} has emerged as an attractive framework due to its simple training objective, stable optimization, and  efficient training. \ac{CFM} defines a tractable probability path conditioned on the target data and derives the associated conditional vector field in closed form. A neural network, hereafter referred to as the velocity network, can therefore be trained to predict this vector field at sampled time coordinates $t\in[0,1]$, where $t$ parameterizes progression along the probability path, without numerically integrating the transport dynamics during training. Consequently, \ac{CFM} has been successfully applied to speech restoration, including \ac{SE}~\cite{cross25a,lee2025flowse,wang2026rethinking}, audio-visual \ac{SE}~\cite{jung2024flowavse}, and target speaker extraction~\cite{hsieh2025adflowtse}.

However, in classical \ac{CFM} the time coordinate $t$ serves two distinct roles: it defines the interpolation  used to construct the conditional probability path and also conditions the velocity network on the corresponding transport state. This temporal conditioning is essential because the target vector field is time-dependent, i.e., it varies along the transport trajectory, reflecting the evolving transport dynamics as the marginal distribution approaches the target distribution. In restoration tasks, however, the source distribution of input signals, which defines the initial state of the generative process at $t=0$, does not provide an equivalent starting point across individual transport paths. Unlike classical \ac{CFM}, where source and target distributions are typically statistically independent (e.g., a standard Gaussian and the distribution of natural images), restoration tasks involve a target signal that is intrinsically contained within the source (e.g., a distorted or noisy target signal).  Because this statistical dependence varies across samples and degradation conditions, two samples with the same $t$ may differ substantially in degradation level and proximity to the target manifold. Consequently, although $t$ remains essential for constructing the conditional probability path during training, it may not unambiguously characterize the current transport state required for conditioning the velocity network. Recent works argue that time-conditioned models must learn trajectory-specific noise scales and may consequently overfit to the prescribed transport path. To address this, they explored time-unconditional formulations by learning a time-independent velocity field without conditioning on $t$~\cite{zhang2026arf, zhang2026generative}. However, removing time conditioning also eliminates an explicit mechanism for adapting the  predicted velocity to the current transport progress. This naturally raises the question of whether transport progress can instead be better inferred from a learned representation of the current sample state.

Representation learning offers a promising direction for overcoming the limitations of existing \ac{CFM}-based restoration frameworks. Latent representations learned through unsupervised and self-supervised reconstruction objectives can capture informative semantic and structural properties of the underlying target signal, as shown by stacked denoising autoencoders~\cite{vincent2010stacked}, context encoders~\cite{pathak2016context}, and masked autoencoders~\cite{he2022masked}. Recent speech restoration methods have further leveraged such representations by combining discriminative and generative models through hybrid architectures or auxiliary conditioning~\cite{abdulatif2024cmgan,li2024diffusion,wang2024gald,zhang2026hyflowse,shetu2026discogan}. Unlike these approaches, which primarily employ discriminative models as auxiliary feature extractors, we view their latent representations as task-relevant encodings that retain information helpful for restoration. We further hypothesize that these representations implicitly encode the sample's restoration state, since successful restoration depends on the relationship between a degraded sample and its target. This motivates the Discriminative Flow-State Hypothesis, which postulates that discriminative representations encode information about the underlying state of a generative transport process. Building upon this hypothesis, we propose \ac{DFM}, which conditions the velocity field on discriminative representations extracted from the iteratively enhanced sample rather than on explicit time $t$, enabling sample-adaptive restoration.

We evaluate \ac{DFM} mainly on \ac{SE} as a representative restoration task. Our analysis shows that discriminative representations capture restoration progress and can be used to estimate the computational effort required for signal enhancement. Leveraging these representations within \ac{CFM}, our proposed \ac{DFM} consistently improves restoration performance while naturally enabling adaptive inference. The main contributions of this work are summarized as follows:

\begin{itemize}

\item We motivate and  formulate the Discriminative Flow-State Hypothesis and demonstrate that discriminative  representations encode generative restoration state and vary consistently with the distributional proximity of samples to the target distribution.

\item We propose \ac{DFM}, a \ac{CFM} framework that conditions the velocity field on discriminative Flow-State representations instead of explicit time conditioning.

\item We validate \ac{DFM} on \ac{SE} and image denoising, demonstrating consistent improvements over  conventional \ac{CFM} and other baselines while naturally supporting adaptive inference with sample-dependent computational budgets.

\end{itemize}

\section{Flow-State Theory}
\label{sec:theory}
The central question of this work is whether representations learned by a discriminative restoration model contain sufficient information to characterize the state of a generative transport process. If so, these representations provide a learned, sample-dependent alternative to the explicit time conditioning for velocity network in conventional \ac{CFM} \cite{lipman2023flow, albergo2022building}.  

\ac{CFM} formulates generative modeling as a continuous transport problem, where a learned velocity field guides probability mass from a source distribution $p_0$ to a target data distribution $p_1$. This formulation constructs a probability path $\{p_t\}_{t \in [0, 1]}$ describing how the marginal distribution evolves along this transport. We define $x_0 \sim p_0$ as the initial state (the source observation) at time $t=0$, and $x_1 \sim p_1$ as the target sample at time $t=1$. For any intermediate time coordinate $t \in (0, 1)$, we further define the intermediate state $x_t \sim p_t$, which is constructed commonly via a linear interpolation path  $x_t = (1 - t)x_0 + t x_1$ and subsequently the velocity network $v_\theta$, parameterized by $\theta$,  is optimized with the
following training objective \cite{lee2025flowse}
\begin{equation}
    \mathcal{L}_{\text{CFM}}(\theta)= \mathbb{E} \left[ \Vert v_\theta(x_t, t,x_0) - (x_1 - x_0) \Vert^2 \right].
\end{equation}
%
%
To quantify generation progress, we define the following:
\begin{definition}[Flow-State]
    We define the Flow-State as the Wasserstein-2 distance between the intermediate marginal distribution $p_t$ and the target data distribution $p_1$
\begin{equation}
\begin{aligned}
    \mathcal{S}(p_t) &\triangleq \Wtwo(p_t, p_1) \\
    &= \left( \inf_{\pi \in \Pi(p_t, p_1)} \mathbb{E}_{(x_t, x_1) \sim \pi} \left[ \|x_t - x_1\|^2 \right] \right)^{1/2},
\label{eq: Wassertein2}
\end{aligned}
\end{equation}
where $\Pi(p_t, p_1)$ denotes the set of all joint distributions with marginals $p_t$ and $p_1$.
\end{definition}
This formulation provides a principled measure of generative progress, as the probability path evolves toward the target distribution, the Wasserstein-2 distance quantifies the remaining distance.

\subsection{Flow-State Properties}
Under standard \ac{OT-CFM}~\cite{lipman2023flow} with independent initialization, e.g., by a standard Gaussian, the target flow trajectory is $x_t = (1-t)x_0 + tx_1$, where $(x_0, x_1) \sim \pi_{\text{OT}}$ and $\pi_{\text{OT}}$ denotes the joint distribution that minimizes the expected transport cost (i.e., the coupling attaining the squared Wasserstein-2 distance $\Wtwo^2(\cdot,\cdot)$)~\cite{mccann1997convexity}.

\begin{proposition}[Properties of Flow-State under \ac{OT-CFM}]
For a transport trajectory governed by $\pi_{\text{OT}}$, the Flow-State $\flowstate(p_t)$ satisfies
\begin{equation}
    \flowstate(p_t) = (1-t) \Wtwo(p_0, \pdata)
    \label{eq:theoremflowstate}
\end{equation}
and $\flowstate(p_t)$ adheres to the following properties $\forall t \in [0, 1]$:
\begin{enumerate}
    \item \textbf{Boundedness:} $0 \le \flowstate(p_t) \le \Wtwo(p_0, \pdata)$.
    \item \textbf{Monotonicity:} $\frac{d}{dt} \flowstate(p_t) = - \Wtwo(p_0, \pdata) < 0$, $\forall t < 1$.
\end{enumerate}
\end{proposition}

\begin{proof}[Proof Sketch] 
The infimum of $\mathbb{E}[\|x_t - x_1\|^2] = (1-t)^2 \mathbb{E}[\|x_0 - x_1\|^2]$ yields the upper bound $\Wtwo^2(p_t, p_1) \le (1-t)^2 \Wtwo^2(p_0, p_1)$. Equality is proven by contradiction, as lower cost 
violates $\pi_{\text{OT}}$'s optimality. Properties 1 and 2 follow from \eqref{eq:theoremflowstate}. See Appendix~A1.1 for the complete proof.
\end{proof}

Proposition 1 shows that, under the OT-CFM conditions, the Flow-State simplifies to a linear scaling of the initial Wasserstein-2 distance over time; hence, using the time $t$ as a surrogate for the Flow-State $\flowstate(p_t)$ is mathematically justified for the standard \ac{OT-CFM} formulation.


\subsection{Ambiguity of Time $t$  in Restoration Tasks}

In restoration tasks, the initial state $x_0$  is typically a mixture of the target signal and a degradation component (e.g., additive noise), with signal scaling and \ac{SNR} varying across samples. In this section, we show that these variations render the global interpolation coordinate $t$ ambiguous as a universal descriptor of the Flow-State $\flowstate(p_t)$,  necessitating a new representation to characterize the restoration progress along the transport trajectory.

\begin{proposition}[Ambiguity of Time $t$]
\label{prop:Ambiguity}
Let the initial state be defined as $x_0^{(\alpha,\beta)} = \alpha x_1 + \beta \nu$, where $\alpha, \beta \in \mathbb{R}^+$ are scaling parameters, $\nu$ is an additive noise signal, and $x_1$ is the target clean signal. For the linear interpolation path $x_t^{(\alpha,\beta)} = (1-t)x_0^{(\alpha,\beta)} + t x_1$ with $t\in[0,1)$, there exists a distinct parameter set $\{\alpha', \beta', t'\}$ with $t \neq t'$ and $(\alpha, \beta) \neq (\alpha', \beta')$ such that
\begin{equation}
p\left(x_t^{(\alpha,\beta)}\right) = p\left(x_{t'}^{(\alpha',\beta')}\right).
\end{equation}

\end{proposition}

\begin{proof}[Proof Sketch]
Expanding the interpolation path yields $x_t^{(\alpha,\beta)} = ((1-t)\alpha + t)x_1 + (1-t)\beta \nu$. For 
$x_t^{(\alpha,\beta)} = x_{t'}^{(\alpha',\beta')}$, matching the coefficients of $x_1$ and $\nu$ gives
\begin{equation}
    \beta' = \frac{1-t}{1-t'}\beta, \quad \alpha' = \frac{(1-t)\alpha + t - t'}{1-t'}.
\end{equation}
For any given $\alpha, \beta > 0$ and $t \in [0,1)$, choosing a distinct $t' \in [0,1)$ such that $t' < (1-t)\alpha + t$ yields valid parameters $\alpha', \beta' > 0$ with $t' \neq t$. Thus, identical states, and, hence, marginal distributions are obtained at different time coordinates. See Appendix~A1.2 for the complete proof.
\end{proof}
This proposition confirms that identical marginal distributions are reached at different times $t$, meaning that the same Flow-State $\flowstate(p_t)$ can be achieved at distinct time steps simply by varying the weighting factors $(\alpha, \beta)$. Consequently, $t$ cannot serve as an unambiguous descriptor of restoration progress. This ambiguity provides the primary motivation for the Discriminative Flow-State Hypothesis, which suggests replacing the time coordinate $t$ with a learned representation of the sample's restoration state.

\begin{proposition}[Upper Bound of Flow-State]
\label{prop:upper_boundmain}
Assuming without loss of generality  that $\mathbb{E}\|x_1\|^2 = \mathbb{E}\|\nu\|^2 = 1$, 
we obtain
\begin{equation}
    \flowstate(p_t) \leq |1-t| (|\alpha-1| + |\beta|)
\end{equation}
and in the special case of a convex combination, for $\beta = 1-\alpha$, where $x_0 = \alpha x_1 + (1-\alpha)\nu$ with $\alpha \in [0, 1]$
\begin{equation}
\Wtwo(p_t, p_1) \leq 2 |1-t| |\alpha-1|.
\end{equation}
\end{proposition}
\begin{proof}[Proof Sketch] We bound $\Wtwo(p_t, p_1)$ via $\mathbb{E}\Vert{}x_t - x_1\Vert{}^2$. Substituting the initialization yields the displacement $\vert{}1-t\vert{}\Vert{}(\alpha-1)x_1 + \beta \nu\Vert{}$. Applying the triangle inequality and the Cauchy-Schwarz
inequality yields the quadratic bound (Full proof is in Appendix A1.3).
\end{proof}
The derived upper bound suggests that $\mathcal{S}(p_t)$ tends to decrease as $t \to 1$, which is consistent with convergence toward the data distribution. The parameter $\alpha$ modulates the trajectory; specifically, larger values of $\alpha$ reduce the bound by positioning $x_t$ closer to $x_1$.

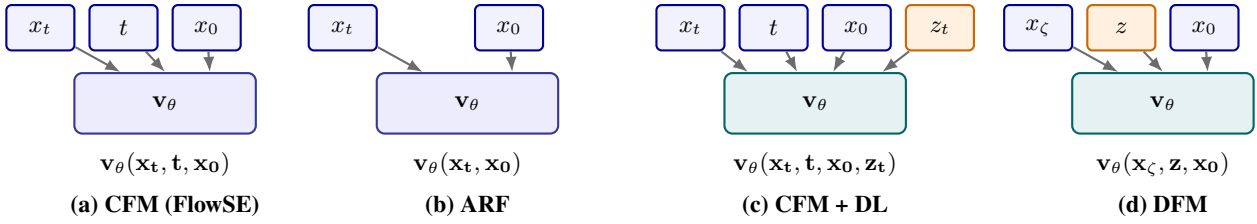
\begin{figure*}[t]
\centering
\begin{tikzpicture}[
font=\small,
input/.style={draw=blue!60!black,fill=blue!5,rounded corners=2pt,thick,minimum width=9mm,minimum height=6mm,align=center},
feature/.style={draw=orange!80!black,fill=orange!12,rounded corners=2pt,thick,minimum width=9mm,minimum height=6mm,align=center},
model/.style={draw=gray!55!blue!80!black,fill=gray!15!blue!8,rounded corners=3pt,thick,minimum width=24mm,minimum height=8mm,align=center},
modelhi/.style={draw=teal!80!black,fill=teal!10,rounded corners=3pt,thick,minimum width=24mm,minimum height=8mm,align=center},
arrow/.style={-{Latex[length=2mm,width=1.5mm]},thick,draw=gray!85!black},
label/.style={font=\small\bfseries,text=black},
title/.style={font=\small\bfseries,text=black}
]

\begin{scope}[xshift=0.6cm]
\node[input] (a1) at (-1.65,0.8) {$x_t$};
\node[input] (a2) at (-0.55,0.8) {$t$};
\node[input] (a3) at (0.55,0.8) {$x_0$};
\node[input,opacity=0] (a4) at (1.65,0.8) {$z$};

\node[model] (am) at (0,-0.2) {$\mathbf{v_\theta}$};

\draw[arrow] (a1)--(am.145);
\draw[arrow] (a2)--(am.north);
\draw[arrow] (a3)--(am.35);

\node[label] at (0,-1.0) {$\mathbf{v_\theta(x_t,t,x_0)}$};
\node[title] at (0,-1.55) {(a) CFM (FlowSE)};
\end{scope}

\begin{scope}[xshift=4.6cm]
\node[input] (b1) at (-1.65,0.8) {$x_t$};
\node[input,opacity=0] (b2) at (-0.55,0.8) {$t$};
\node[input] (b3) at (0.55,0.8) {$x_0$};
\node[input,opacity=0] (b4) at (1.65,0.8) {$z$};

\node[model] (bm) at (0,-0.2) {$\mathbf{v_\theta}$};

\draw[arrow] (b1)--(bm.145);
\draw[arrow] (b3)--(bm.35);

\node[label] at (0,-1.0) {$\mathbf{v_\theta(x_t,x_0)}$};
\node[title] at (0,-1.55) {(b) ARF};
\end{scope}

\begin{scope}[xshift=9.2cm]
\node[input] (c1) at (-1.65,0.8) {$x_t$};
\node[input] (c2) at (-0.55,0.8) {$t$};
\node[input] (c3) at (0.55,0.8) {$x_0$};
\node[feature] (c4) at (1.65,0.8) {$z_t$};

\node[modelhi] (cm) at (0,-0.2) {$\mathbf{v_\theta}$};

\draw[arrow] (c1)--(cm.155);
\draw[arrow] (c2)--(cm.120);
\draw[arrow] (c3)--(cm.60);
\draw[arrow] (c4)--(cm.25);

\node[label] at (0,-1.0) {$\mathbf{v_\theta(x_t,t,x_0,z_t)}$};
\node[title] at (0,-1.55) {(c) CFM + DL};
\end{scope}

\begin{scope}[xshift=13.8cm]
\node[input] (d1) at (-1.65,0.8) {$x_{\zeta}$};
\node[feature] (d2) at (-0.55,0.8) {$z$};
\node[input] (d3) at (0.55,0.8) {$x_0$};
\node[input,opacity=0] (d4) at (1.65,0.8) {$\zeta$};

\node[modelhi] (dm) at (0,-0.2) {$\mathbf{v_\theta}$};

\draw[arrow] (d1)--(dm.145);
\draw[arrow] (d2)--(dm.north);
\draw[arrow] (d3)--(dm.35);

\node[label] at (0,-1.0) {$\mathbf{v_\theta(x_{\zeta},z,x_0)}$};
\node[title] at (0,-1.55) {(d) DFM};
\end{scope}

\end{tikzpicture}

\caption{Comparison of conditioning strategies for Flow-based restoration tasks. (a) \ac{CFM} conditions the velocity field on $t$. (b) Autonomous Rectified Flow (ARF) removes $t$. (c) CFM+DL conditions CFM with the additional discriminative latent (DL) representation $z_t=\mathcal{E}_{\phi}(x_t)$. (d) \ac{DFM} replaces $t$ with the Discriminative Flow-State Representation $z=\mathcal{E}_{\phi}(x_{\zeta}).$}

\label{fig:conditioning_comparison}
\end{figure*}

\subsection{Discriminative Flow-State Hypothesis}

Let $\mathcal{X}$ denote the data space, and let $\mathcal{E}_{\phi}:\mathcal{X}\rightarrow\mathcal{Z}$ be a pretrained, frozen discriminative encoder parameterized by $\phi$, mapping a sample $x\in \mathcal{X}$ to a latent representation $z=\mathcal{E}_{\phi}(x)$ within the representation space $\mathcal{Z}$. Furthermore, let  $\hat{x}=\mathcal{D}_\Phi(z) \sim p_d$, where $p_d$ represents the output distribution of the discriminative model $\mathcal{F}_{\Omega}= \mathcal{D}_\Phi \circ  \mathcal{E}_{\phi}$ parameterized by $\Omega$, and $\mathcal{D}_{\Phi}$ denotes the pretrained discriminative decoder parameterized by $\Phi$.

\begin{hypothesis}[Discriminative Flow-State Hypothesis]
\label{hyp:DFlow}
Latent representations $z\in \mathcal{Z}$ learned by discriminative restoration models encode information about the underlying Flow-State $\mathcal{S}(p_t)$ governing generative transport.
\end{hypothesis}
Consequently, rather than relying explicitly on time $t$, discriminative latent representations $z$ provide a sample-dependent description of restoration progress, proximity to the target distribution, and restoration complexity. 
Furthermore, evaluating the true Flow-State $\mathcal{S}(p_t) = \Wtwo(p_t, p_1)$ is infeasible in practice because the ground-truth clean data distribution $p_1$ is inaccessible at inference. Instead, we define the approximate Flow-State $\hat{\mathcal{S}}(p_t) := \Wtwo(p_t, p_d)$ as an empirical surrogate, as we can observe samples from the distribution $p_d$. To justify this choice, the following proposition shows that the approximation error is bounded by the discriminative model's mean squared reconstruction error loss.

\begin{proposition}[Approximation of Flow-State]
\label{prop:discriminative_boundedness}
Let $\hat{x}_1\sim p_d$ be the estimate of $x_1\sim p_1$ obtained from a discriminative restoration model trained under a mean squared error objective $\mathcal{L} = \mathbb{E}[\Vert{}\hat{x}_1 - x_1\Vert{}^2]$. 
Then, we have
\begin{equation}
    |\mathcal{S}(p_t) - \hat{\mathcal{S}}(p_t)|^2 \le \mathcal{L}.
\end{equation}
\end{proposition}

\begin{proposition}[Flow-State Preservation]
\label{prop:flow_state_preservation}
Assuming the discriminative model's architecture forms a Markov chain $x_t \to z \to \hat{x}_1$, where $z = \mathcal{E}_{\phi}(x_t)$ 
and $\hat{x}_1$ is the estimated clean signal. Then,
\begin{equation}
    \mathcal{I}(z;x_1)\geq \mathcal{I}(\hat{x}_1;x_1),
\end{equation}
where $\mathcal{I}$ denotes the mutual information.
\end{proposition}


The proofs are provided in Appendix~A1.4 and~A1.5. Propositions~\ref{prop:discriminative_boundedness} and~\ref{prop:flow_state_preservation} jointly justify our \ac{DFM} framework. Specifically, Proposition~\ref{prop:discriminative_boundedness} ensures the deviation between $\mathcal{S}(p_t)$ and its approximation $\hat{\mathcal{S}}(p_t)$ shrinks as discriminative model's training loss is minimized, while Proposition~\ref{prop:flow_state_preservation} shows that $z$ retains at least as much target information as $\hat{x}_1$. Hence, leveraging $z$ rather than $\hat{x}_1$ for characterizing transport progress is justified. The assumption of the discriminative model forming a Markov chain is admittedly a strong one, but would be fulfilled, e.g., for an autoencoder model without skip connections. 
However, Sections~\ref{res:Flow-State Preservation}--\ref{sec:adaptive_inference} empirically validate that discriminative latent representations can be utilized for accurately predicting different observable proxies of the Flow-State. In the remainder of this paper, we define these discriminative latent representations $z$ as the \ac{DFSR}.


\section{Discriminative Flow Matching}
\label{sec:method}

Following the additive signal model introduced in Proposition~\ref{prop:Ambiguity}, we describe the observed data as
\begin{equation}
    x_0 = \alpha x_1 + \beta \nu.
\label{eq:additive}
\end{equation}
To this end, motivated by the Discriminative Flow-State Hypothesis, we propose \ac{DFM}, which conditions the velocity
field on \ac{DFSR} $z=\mathcal{E}_{\phi}(x_{\zeta}),
$ derived from a  pretrained encoder $\mathcal{E}_{\phi}$, instead of time $t$.
Following~\cite{lee2025flowse, richter2023speech}, we define the conditional probability path $p(x_\zeta \mid x_0, x_1) = \mathcal{N}(x_\zeta; \mu_\zeta, \sigma_\zeta^2 \mathbf{I})$ mapping the Gaussian-perturbed observation $x_0 + \sigma\epsilon$ at $\zeta = 0$ to the clean target $x_1$ at $\zeta = 1$:
\begin{equation}
\mu_{\zeta} = (1-\zeta)x_0 + \zeta x_1,
\quad
x_{\zeta} = \mu_{\zeta} + \sigma_\zeta\epsilon,
\label{eq:path}
\end{equation}
where DFM  interpolation coordinate  $\zeta \sim \mathcal{U}[0,1]$, $\epsilon \sim \mathcal{N}(0, \mathbf{I})$, and
$\sigma_{\zeta}=(1-\zeta)\sigma$ for a constant $\sigma\geq0$ . With the corresponding target velocity $v^\star = \frac{\partial x_{\zeta}}{\partial \zeta} = (x_1 - x_0) - \sigma\epsilon$, the \ac{DFM} training objective is formulated as
\begin{equation}
\mathcal{L}_{\mathrm{DFM}}(\theta) = \mathbb{E} \left[ \left\| v_{\theta}(x_{\zeta}, z,x_0) - v^\star \right\|_2^2 \right].
\label{eq:dfm}
\end{equation}

\noindent \textbf{DFM Inference:} Inference is initialized from the degraded observation, $x^{(0)}=x_0+\sigma\epsilon$,
and performed using Euler integration
\begin{equation}
x^{(n+1)}
= x^{(n)} +\gamma_n v_\theta\left( x^{(n)},
\mathcal{E}_{\phi}(x^{(n)}), x_0 \right),
\label{eq:euler}
\end{equation}
where $\gamma_n$ denotes the integration step size at $n=0,\ldots,N-1$ and $N$ is the total \ac{NFE}. For simplicity, we use a uniform time discretization, i.e., $\gamma_n=1/N$ $\forall n$.
The estimate of the clean data is obtained as $\hat{x}_1=x^{(N)}$. Since the velocity network $v_\theta$ is not dependent on the interpolation factor $\zeta$, transport is guided entirely by the evolving \ac{DFSR} $z$, which can also be used to adapt either the computational budget $N$ or the integration step sizes $\gamma_n$ as described in Section~\ref{sec:adaptive_inference}.

\section{Experiments and Results}
\label{sec:experiments}

\subsection{Flow-State Preservation}
\label{res:Flow-State Preservation}

We first validate the Discriminative Flow-State Hypothesis by examining whether \ac{DFSR} preserve information about the underlying Flow-State. Figure~\ref{t-SNE dccrn} visualizes latent representations extracted from a pretrained DCCRN encoder \cite{hu2020dccrn} across SNRs ranging from $-30$ to $40$~dB. The latent space is organized according to degradation severity, with progressively cleaner samples placed closer to the clean latent manifold.

\begin{figure}[t]
\centering
\includegraphics[width=0.9\columnwidth]{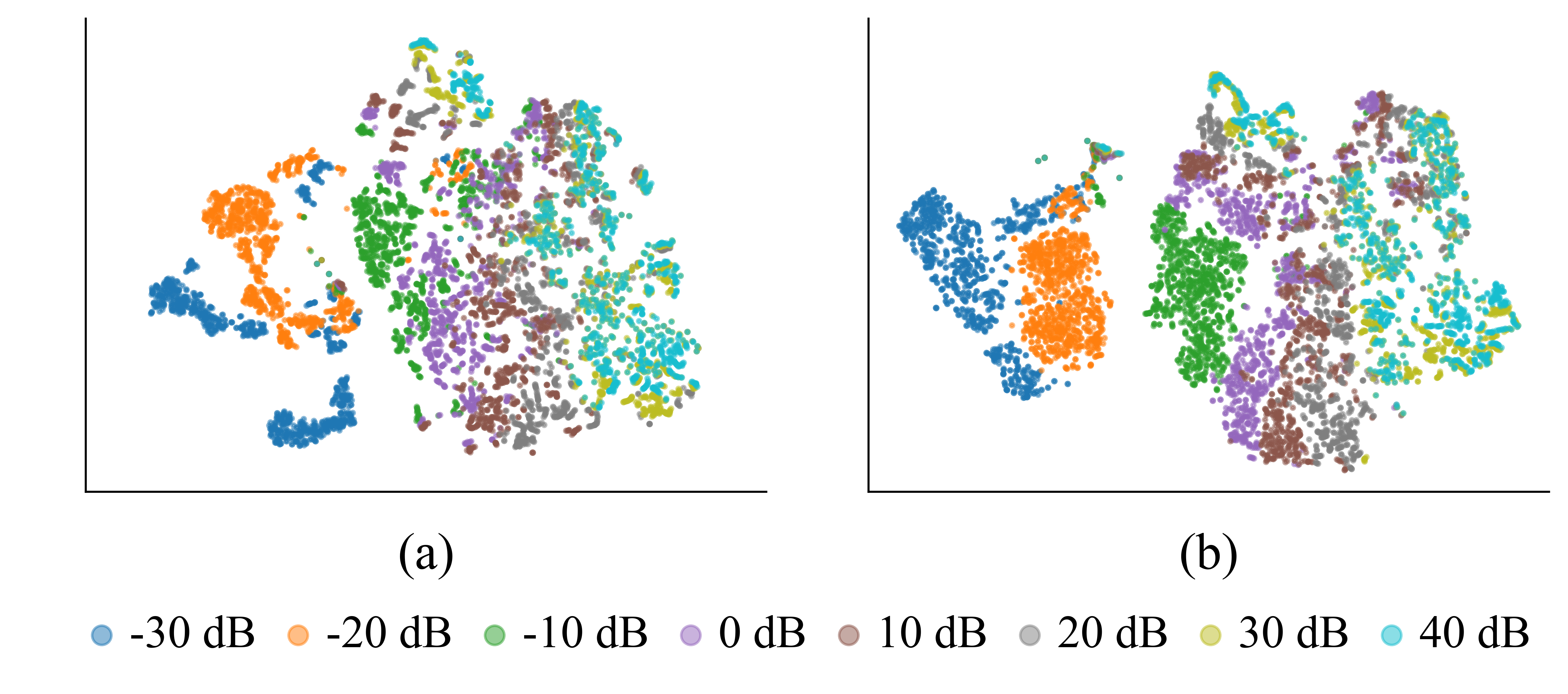}
\caption{t-SNE visualization of DCCRN latent representations under (a) rain and (b) machinery noise.}
\label{t-SNE dccrn}
\end{figure}

We further predict three quantities related to the Flow-State from the \ac{DFSR}: the interpolation factor used in the training phase $\zeta$, the ground truth intermediate SNR (iSNR) computed from the mean $\mu_{\zeta}$ (see \eqref{eq:path}), and global (input) SNR (gSNR) computed from the degraded observation $x_0$. Both \ac{SNR} measures are normalized to $[0,1]$ as

\begin{equation}
    s_{\mathrm{SNR}} \triangleq \max\left(0, \min\left(1, \frac{\mathrm{SNR} + 10}{70}\right)\right),
\end{equation}
mapping the range $[-10,60]$~dB to $[0,1]$. 

Figure~\ref{fig:state_prediction} compares our proposed latent-predictor against an anchor-based baseline. The latent-predictor utilizes the \ac{DFSR} $z_\mu = \mathcal{E}_{\phi}(\mu_{\zeta})$ as input. Conversely, the anchor-based baseline uses the signal pair $\{\mu_{\zeta}, \mathcal{F}_{\Omega}(\mu_{\zeta})\}$. Both architectures utilize the PESQNet backbone~\cite{xu2022pesqnet}, and are trained on the DNS Challenge dataset~\cite{reddy2020interspeech} and evaluated on $1,500$ simulated samples.


\begin{figure}[t]
\centering

\begin{tikzpicture}[
font=\small,
block/.style={draw,rounded corners=2pt,minimum width=1.45cm,minimum height=0.48cm,align=center},
pred/.style={draw,fill=gray!10,rounded corners=2pt,minimum width=1.45cm,minimum height=0.58cm,align=center},
arrow/.style={-{Latex[length=1.4mm]},thick}
]
\node[block] (z) at (0,1.0) {$z=\mathcal{E}_{\phi}(\mu_{\zeta})$};
\node at (0,0.35) {\textbf{or}};
\node[block] (anchor) at (0,-0.3) {$\{\mu_{\zeta},\mathcal{F}_{\Omega}(\mu_{\zeta})\}$};
\node[block] (input) at (2.3,0.35) {Input};

\draw[arrow] (z)--(input);
\draw[arrow] (anchor)--(input);

\node[pred] (p1) at (4.7,1.25) {Predictor};
\node[pred] (p2) at (4.7,0.35) {Predictor};
\node[pred] (p3) at (4.7,-0.55) {Predictor};

\draw[arrow] (input)--(p1);
\draw[arrow] (input)--(p2);
\draw[arrow] (input)--(p3);

\node[block] (o1) at (6.9,1.25) {$\widehat{\zeta}$};
\node[block] (o2) at (6.9,0.35) {$\widehat{\text{iSNR}}$};
\node[block] (o3) at (6.9,-0.55) {$\widehat{\text{gSNR}}$};

\draw[arrow] (p1)--(o1);
\draw[arrow] (p2)--(o2);
\draw[arrow] (p3)--(o3);
\end{tikzpicture}
\caption{Prediction of observable Flow-State quantities from either the \ac{DFSR} or the anchor-based signal pair.}
\label{fig:state_prediction}
\end{figure}
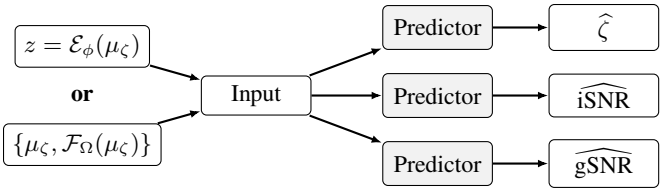

\begin{table}[t]
\centering
\begin{tabular}{llccc}
\toprule
Target & Method & MAE $\downarrow$ & RMSE $\downarrow$ & Pearson $\uparrow$\\
\midrule
iSNR
& Latent & \textbf{0.032} & \textbf{0.047} & 0.968\\
& Anchor & 0.035 & 0.048 & \textbf{0.971}\\
\midrule
gSNR
& Latent & \textbf{0.117} & \textbf{0.154} & \textbf{0.491}\\
& Anchor & 0.142 & 0.181 & 0.301\\
\midrule
$\zeta$
& Latent & \textbf{0.154} & \textbf{0.190} & 0.741\\
& Anchor & 0.155 & 0.211 & \textbf{0.751}\\
\bottomrule
\end{tabular}
\caption{Prediction performance for quantities related to the Flow-State. Best results are shown in \textbf{bold}.}
\label{tab:target_predictions}
\end{table}

Table~\ref{tab:target_predictions} shows that iSNR predictors can estimate iSNR accurately irrespective of the used input features, achieving Pearson correlations above $0.96$.  Notably, the latent-based predictor achieves the lowest MAE and RMSE across all three targets, while Pearson correlations are comparable between the two input types. As noted in Proposition~2, estimating the interpolation factor $\zeta$ and gSNR is inherently challenging, as varying parameterizations of $\alpha$ and $\beta$  yield identical states for $\mu_{\zeta}$. Consequently, as confirmed by the empirical results, even the anchor-based predictor—which also incorporates $\mu_{\zeta}$ as input feature, struggles to accurately estimate $\zeta$ or gSNR. Nevertheless, the latent-based predictor's strong performance confirms that \ac{DFSR} effectively preserve Flow-State information, empirically validating Proposition~5 and the Discriminative Flow-State Hypothesis.

\begin{figure}[t]
\centering
\includegraphics[width=0.85\columnwidth]{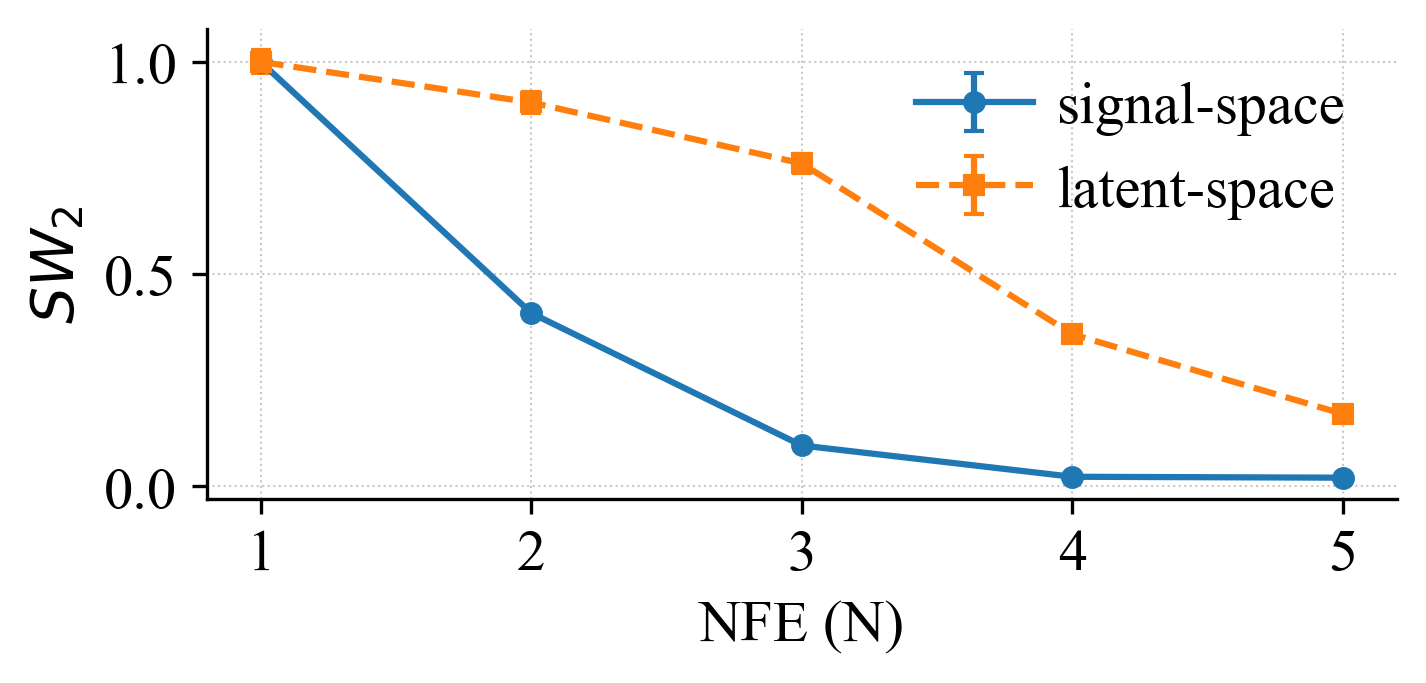}
\caption{Evolution of signal-space and latent-space Flow-States over NFE.}
\label{fig:latent_geometry}
\end{figure}

\subsection{Flow-State Monotonicity}
\label{sec:monotonicity}

We further analyze  whether successive \ac{DFM} updates induce a monotonic decrease of the Flow-State of the distribution $p_n$ of the intermediate estimates $x^{(n)}$ (see \eqref{eq:euler}) of \ac{DFM} toward the clean-speech distribution $\pdata$. The Flow-State in signal-space at the $n$-th inference step can be defined as
\begin{equation}
\flowstate(p_n)=\Wtwo(p_n,\pdata),
\label{eq:SW distance}
\end{equation}
analogously to Definition~1 in Section~\ref{sec:theory}. Since exact computation of $\Wtwo$ is computationally prohibitive in high-dimensional waveform spaces, we estimate it using the Sliced Wasserstein-2 distance ($\mathcal{SW}_2$)~\cite{rabin2011wasserstein,kolouri2019generalized}. Empirical signal-space distributions are constructed from overlapping $25$-ms waveform segments extracted from the DNS Challenge non-reverberant test set~\cite{reddy2020interspeech}, consisting of $150$ samples with diverse noise conditions. Similar to \eqref{eq:SW distance}, we compute $\mathcal{SW}_2$ in the latent-space, where we obtain the corresponding discriminative latent representation of the samples of $p_n$ and $\pdata$ from the frozen DCCRN encoder~\cite{hu2020dccrn}.

Distances are estimated using $512$ random projections and averaged over $20$ trials, with features standardized via clean-data statistics. As shown in Fig.~\ref{fig:latent_geometry}, both signal- and latent-space Flow-States decrease monotonically during inference. The signal-space trajectory achieves Spearman's $\rho = -1.000$, Pearson's $r = -0.888$, and a convergence floor of $0.112\pm0.005$. The latent-space trajectory exhibits identical monotonic ordering ($\rho = -1.000$), a stronger Pearson correlation ($r = -0.973$), and zero monotonicity violations. The signal- and latent-space trajectories are rank-equivalent ($\rho = 1.000$), indicating that the \ac{DFSR} preserves the underlying Flow-State ordering. These results indicate that successive \ac{DFM} updates monotonically reduce the Flow-State, mirroring the behavior established analytically for OT-CFM in Proposition~1, and that the latent-space trajectory follows the same ordering, which supports the Discriminative Flow-State Hypothesis.


\subsection{Speech Enhancement}
\label{sec:seexp}

\subsubsection{Dataset}
We evaluate the proposed method on the Interspeech 2020 DNS Challenge dataset~\cite{reddy2020interspeech}, a widely adopted benchmark for supervised \ac{SE}. The training set contains approximately $1{,}000$ hours of noisy-clean speech pairs generated by mixing clean speech with noise at SNRs between $-10$ and $30$~dB, where $50\%$ of the utterances are additionally reverberated using the DNS Challenge RIR corpus~\cite{reddy2020interspeech}. Following the official evaluation protocol, performance is reported on both the synthetic non-reverberant test set ($150$ utterances), which enables reference-based evaluation, and the DNS Real Recordings test set ($300$ recordings), which assesses generalization to real-world acoustic conditions.

\subsubsection{Implementation Details}

Unless stated otherwise, all \ac{CFM} methods are trained using the FlowSE training framework~\cite{lee2025flowse} and the NCSN++ backbone~\cite{song2021score}, which is widely adopted in the literature for \ac{CFM} and diffusion-based \ac{SE}. \ac{DFM} employs a frozen DCCRN~\cite{hu2020dccrn} to extract \ac{DFSR} after the recurrent layers, projecting them into a $256$-dimensional embedding that is linearly interpolated to match the backbone's temporal resolution. This representation replaces the standard time embedding via additive conditioning~\cite{lee2025flowse,richter2023speech}. All models are trained using the Adam optimizer with a learning rate of $10^{-4}$, a batch size of $8$, and an exponential moving average decay of $0.999$. Following~\cite{lee2025flowse}, we use an FFT length of $510$, a window length of $510$, a hop length of $128$, and a Gaussian noise prefactor $\sigma = 0.487$, which was shown to be effective for the \ac{SE} tasks. Models are trained on a single A100 GPU for $800{,}000$ steps, with the best checkpoints selected via validation PESQ and SI-SDR, and results reported as the average over $5$ independent runs.

\begin{table}[t]
\centering
\small
\setlength{\tabcolsep}{1.8pt} 
\begin{tabular}{llcccc}
\toprule
\textbf{Cat.} & \textbf{Method} & \begin{tabular}[c]{@{}c@{}}\textbf{SI-SDR}\\($\uparrow$)\end{tabular} & \begin{tabular}[c]{@{}c@{}}\textbf{PESQ}\\($\uparrow$)\end{tabular} & \begin{tabular}[c]{@{}c@{}}\textbf{SCOREQ}\\ \textbf{Ref.} ($\downarrow$)\end{tabular} & \begin{tabular}[c]{@{}c@{}}\textbf{SCOREQ}\\ \textbf{Non-Ref.}($\uparrow$)\end{tabular} \\
\midrule
Unproc. & Noisy & 9.06 & 1.58 & 0.93 & 2.78 \\
& Clean & -- & -- & 0.00 & 4.60\\
\midrule
Disc. & DCCRN & 17.36 & 2.91 & 0.31 & 4.15\\
& GCRN & 16.71 & 2.63 & 0.42 & 3.87 \\
\midrule
Diff. & SGMSE+ & 16.86 & 2.81 & 0.29 & 4.15\\
& BBED  & 19.10 & 2.81 & 0.26 & 4.38 \\
& GALDSE & 18.04 & 2.77 & 0.30 & 4.23 \\
& SEBridge & 17.21 & 2.45 & 0.38 & 4.08 \\
& SToRM & 17.56 & 2.80 & 0.27 & 4.33 \\
\midrule
CFM & FlowSE & 18.99 & 2.86 & 0.25 & 4.44\\
& ARF & 19.04 & 2.82& 0.25 &4.43\\
& CFM+DL & 19.47 & 2.87& 0.25& 4.43 \\
& \textbf{DFM (Ours)} & \textbf{19.63} & \textbf{2.96} & \textbf{0.23} & \textbf{4.51} \\
\bottomrule
\end{tabular}
\caption{\ac{SE} results on the DNS Challenge synthetic non-reverberant test set.}
\label{tab:DNSnoreverb1}
\end{table}

\begin{table}[t]
\centering
\small
\setlength{\tabcolsep}{6.5pt} 
\begin{tabular}{lcccc}
\toprule
\textbf{Method} & \begin{tabular}[c]{@{}c@{}}\textbf{P808 MOS}\\($\uparrow$)\end{tabular} & \begin{tabular}[c]{@{}c@{}}\textbf{SIG}\\($\uparrow$)\end{tabular} & \begin{tabular}[c]{@{}c@{}}\textbf{BAK}\\($\uparrow$)\end{tabular} & \begin{tabular}[c]{@{}c@{}}\textbf{OVRL}\\($\uparrow$)\end{tabular} \\
\midrule
Noisy & 3.05 & 3.05 & 2.51 & 2.26 \\
RVAE  & 3.29 & 3.16 & 2.91 & 2.44 \\
CDiffuse  & 3.14 & 3.15 & 3.19 & 2.55 \\
MetricGAN+ & 3.26 & 2.88 & 3.39 & 2.45 \\
DisCoGAN & 3.65 & 3.32 & 3.91 & 2.98 \\
SGMSE  & 3.38 & 3.22 & 3.02 & 2.52 \\
SGMSE+  & 3.64 & 3.42 & 3.82 & 3.04 \\
\textbf{DFM (Ours)} & \textbf{3.75} & \textbf{3.43} & \textbf{4.03} & \textbf{3.14} \\
\bottomrule
\end{tabular}
\caption{Generalization results on the DNS Real Recordings test set~\cite{reddy2020interspeech}. Baseline results are taken from their respective original publications.}
\label{tab:DNSReal}
\end{table}

\begin{table}[ht]
\centering
\small
\setlength{\tabcolsep}{4.5pt} 
\begin{tabular}{lcccc}
\toprule
\textbf{Method} & \begin{tabular}[c]{@{}c@{}}\textbf{SI-SDR}\\($\uparrow$)\end{tabular} & \begin{tabular}[c]{@{}c@{}}\textbf{PESQ}\\($\uparrow$)\end{tabular} & \begin{tabular}[c]{@{}c@{}}\textbf{SCOREQ}\\ \textbf{Ref.} ($\downarrow$)\end{tabular} & \begin{tabular}[c]{@{}c@{}}\textbf{SCOREQ}\\ \textbf{Non-Ref.}($\uparrow$)\end{tabular} \\
\midrule
Noisy & 9.06 & 1.58 & 0.93 & 2.78 \\
Clean & -- & -- & 0 & 4.60\\
\midrule
CFM  & 18.74 & 2.78 & 0.27 & 4.32\\
ARF & 18.83 & 2.79& 0.26&4.38\\
CFM+DL & 18.96& 2.75& 0.26& 4.39 \\
\textbf{DFM (Ours)} & \textbf{19.17} & \textbf{2.88} & \textbf{0.24} & \textbf{4.44} \\
\bottomrule
\end{tabular}
\caption{Results on the DNS Challenge non-reverberant test set under the data-prediction training objective. }
\label{tab:DNSnoreverb2}
\end{table}

\subsubsection{Experimental Results}

To isolate the effect of conditioning on \ac{DFSR} $z$ compared to time conditioning, we evaluate \ac{DFM} against three  re-trained baseline \ac{CFM} methods, CFM (FlowSE) \cite{lee2025flowse}, \ac{CFM}+DL, and ARF~\cite{zhang2026arf} as shown in Fig.~\ref{fig:conditioning_comparison},  using identical backbones, trajectories, and training settings. We further benchmark against state-of-the-art discriminative, GAN-, and diffusion-based models: DCCRN~\cite{hu2020dccrn}, GCRN~\cite{tan2019learning}, SGMSE~\cite{welker2022speech}, SGMSE+~\cite{richter2023speech}, BBED~\cite{lay2023reducing}, GALDSE~\cite{wang2024gald}, SEBridge~\cite{qiu2023se}, SToRM~\cite{lemercier2023storm}, RVAE~\cite{leglaive2020recurrent}, CDiffuse~\cite{lu2022conditional}, MetricGAN+~\cite{fu2021metricgan+}, and DisCoGAN~\cite{shetu2025gan}. Performance is evaluated using standard \ac{SE} metrics: SI-SDR~\cite{le2019sdr} for signal reconstruction, wide-band PESQ~\cite{rix2001perceptual} for perceptual quality, SCOREQ~\cite{ragano2024scoreq} for DNN-based reference and non-reference based assessment, and DNSMOS~\cite{reddy2021dnsmos} for non-intrusive evaluation on real recordings.

Table~\ref{tab:DNSnoreverb1} shows results on the DNS synthetic non-reverberant test set. \ac{DFM} outperforms all \ac{CFM} variants, achieving $19.63$ dB SI-SDR and $2.96$ PESQ, gaining $0.64$ dB in SI-SDR and $0.10$ in PESQ over the baseline  \ac{CFM}. Relative to the discriminative baseline DCCRN, \ac{DFM} improves SI-SDR by more than $2.2$~dB and reduces the reference SCOREQ distance from $0.31$ to $0.23$. These results indicate that replacing  $t$  with \ac{DFSR} $z$ leads to more effective generative transport. Table~\ref{tab:DNSReal} evaluates generalization on the DNS Real Recordings test set, where \ac{DFM} outperforms the baseline methods, obtaining a P.808 MOS of $3.75$, together with best SIG ($3.43$), BAK ($4.03$), and OVRL ($3.14$) scores, demonstrating strong generalization to real-world recordings. We further evaluate \ac{DFM} using alternative discriminative backbones, including GCRN, TaylorSENet, and ULCNet~\cite{tan2019learning, ijcai2022p582, shetu2023ultra}. As shown in Appendix~A1.8, \ac{DFM} generalizes robustly across these architectures, consistently outperforming the baselines. Finally, Table~\ref{tab:DNSnoreverb2} reports results under the data-prediction objective ($\mathcal{L}_{\mathrm{DFM}}(\theta) = \mathbb{E}\left[\Vert{}v_\theta(x_{\zeta}, z, x_0) - x_1\Vert{}_2^2\right]$). \ac{DFM} consistently outperforms CFM, CFM+DL, and ARF across all evaluation metrics. The consistent improvements under both training objectives, together with the robust generalization to synthetic and real-world test sets, demonstrate the effectiveness of the proposed \ac{DFM} framework.

\begin{figure}[t]
\centering
\includegraphics[width=\columnwidth]{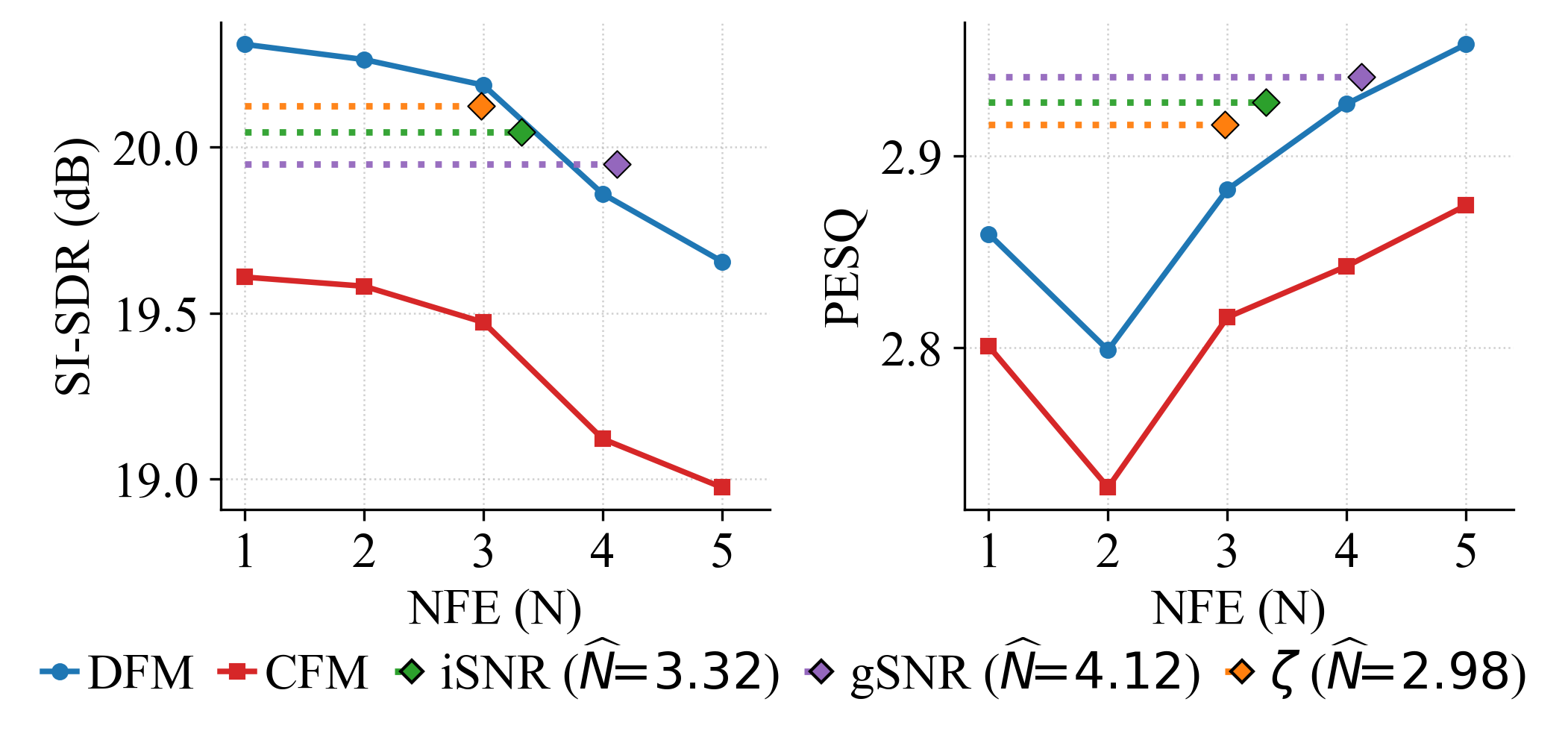} %
\caption{Performance for DFM with adaptive inference.}
\label{fig:Ada Performance}
\end{figure}

\subsubsection{Adaptive Inference}
\label{sec:adaptive_inference}

We hypothesize that the Flow-State characterizes not only restoration progress but also the sample-specific computational effort required for inference. Samples closer to the target distribution should require fewer \acp{NFE}, whereas challenging samples benefit from additional refinement. We define restoration complexity as the minimum number of function evaluations $N^\star(x_0)$ required to keep the reconstruction error $\mathcal{L}_{rec}$ (e.g., PESQ) below a threshold $\xi > 0$
\begin{equation}
    N^\star(x) = \min \left\{ N \in \mathbb{N} : \mathcal{L}_{rec}(\hat{x}_1^{(N)}, x_1) \leq \xi \right\}.   
    \label{eq:optimalN}
\end{equation}
Motivated by Section~\ref{res:Flow-State Preservation}, we introduce a predictor $\mathcal{P}_{\omega}$ that leverages Flow-State estimates (gSNR, iSNR, and $\zeta$) as a sample difficulty estimate to dynamically determine the computational budget and integration schedule. Given the initial representation $z^{(0)} = \mathcal{E}_{\phi}(x_0)$, the predictor jointly outputs the required \ac{NFE} $\widehat{N} \le N_{\max}$ and step sizes $\boldsymbol{\gamma} = [\gamma_1, \ldots, \gamma_{\widehat{N}}]$
\begin{equation}
(\widehat{N}, \boldsymbol{\gamma}) = \mathcal{P}_{\omega}(z^{(0)}),
\end{equation}
where $N_{\max}$ is the maximum allowed \ac{NFE} budget. This yields a sequence of non-uniform, sample-dependent step sizes assigning larger updates early in transport and smaller steps for later refinement (see Appendix~A1.12 for implementation details). Updates are performed according to \eqref{eq:euler}.
%

Fig.~\ref{fig:Ada Performance} evaluates adaptive inference performance by comparing standard fixed-budget \ac{CFM} and \ac{DFM} baselines ($N_{\max} = 5$) against adaptive \ac{DFM} variants  with gSNR, iSNR, and  $\zeta$ predictors. All adaptive strategies require significantly fewer \ac{NFE} than the maximum budget, achieving average \ac{NFE} of $4.12$, $3.32$, and $2.98$, respectively.  Despite this reduced computation, all adaptive variants outperform the \ac{CFM} baselines and surpass fixed-budget  \ac{DFM} in SI-SDR at $N=5$ while retaining comparable PESQ. We observe a nuanced trade-off: while SI-SDR is maximized with fewer \ac{NFE}, PESQ benefits from additional refinement. Notably, the $\zeta$-based predictor achieves $20.12$ dB SI-SDR and $2.92$ PESQ at $2.98$ NFE, improving SI-SDR over fixed-budget DFM by 0.49 dB at a 40.4\% reduction in computation relative to $N_{\max} = 5$.  These results confirm that \ac{DFM}, by conditioning on $z$, effectively aligns computational effort with sample-specific complexity to ensure a favorable quality-computation trade-off.
\begin{figure}[t]
\centering
\includegraphics[width=0.85\columnwidth]{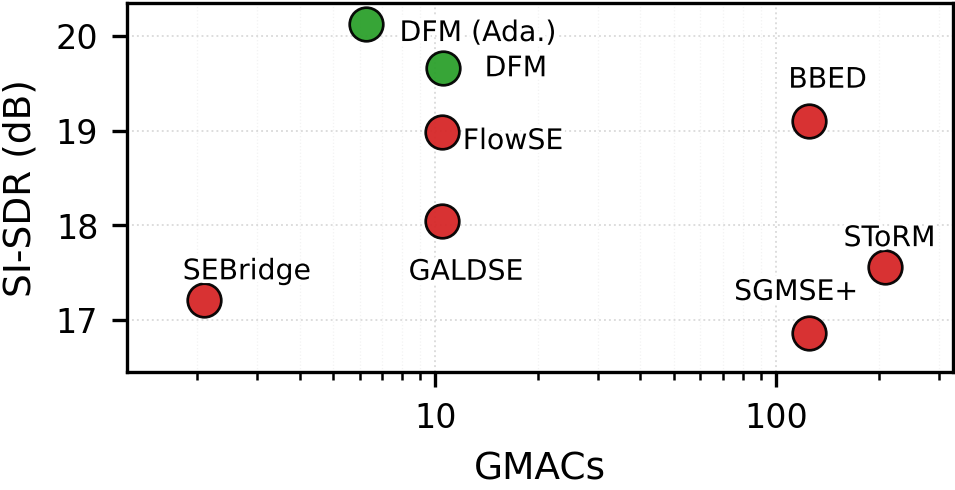} %
\caption{Complexity-performance trade-off (baselines shown in red, proposed methods in green)}
\label{fig:ComplexityPerformance}
\end{figure}

\subsubsection{Complexity--Performance Trade-off}

\ac{DFM} introduces additional computation via the discriminative encoder $\mathcal{E}_{\phi}$ (the DCCRN encoder, which contains approximately $2.6$M parameters). To offset this overhead, we reduce the conditioning dimension from the original $512$-dimensional time embedding used in FlowSE \cite{lee2025flowse} to a $256$-dimensional  representation. Consequently, the backbone is reduced to $63.2$M parameters, yielding a total model size of $65.8$M parameters, comparable to the FlowSE baseline ($65.6$M).

Fig.~\ref{fig:ComplexityPerformance} illustrates the complexity--performance trade-off of \ac{DFM} against recent diffusion- and Flow-based \ac{SE} methods in terms of SI-SDR and GMACs per frame. \ac{DFM} achieves superior enhancement performance while maintaining competitive computational complexity. Furthermore, adaptive inference reduces the average computational cost to $6.27$ GMACs per frame, outperforming both fixed-budget \ac{DFM} and existing diffusion- and Flow-based approaches.

\begin{figure}[t]
\centering
\includegraphics[width=\columnwidth]{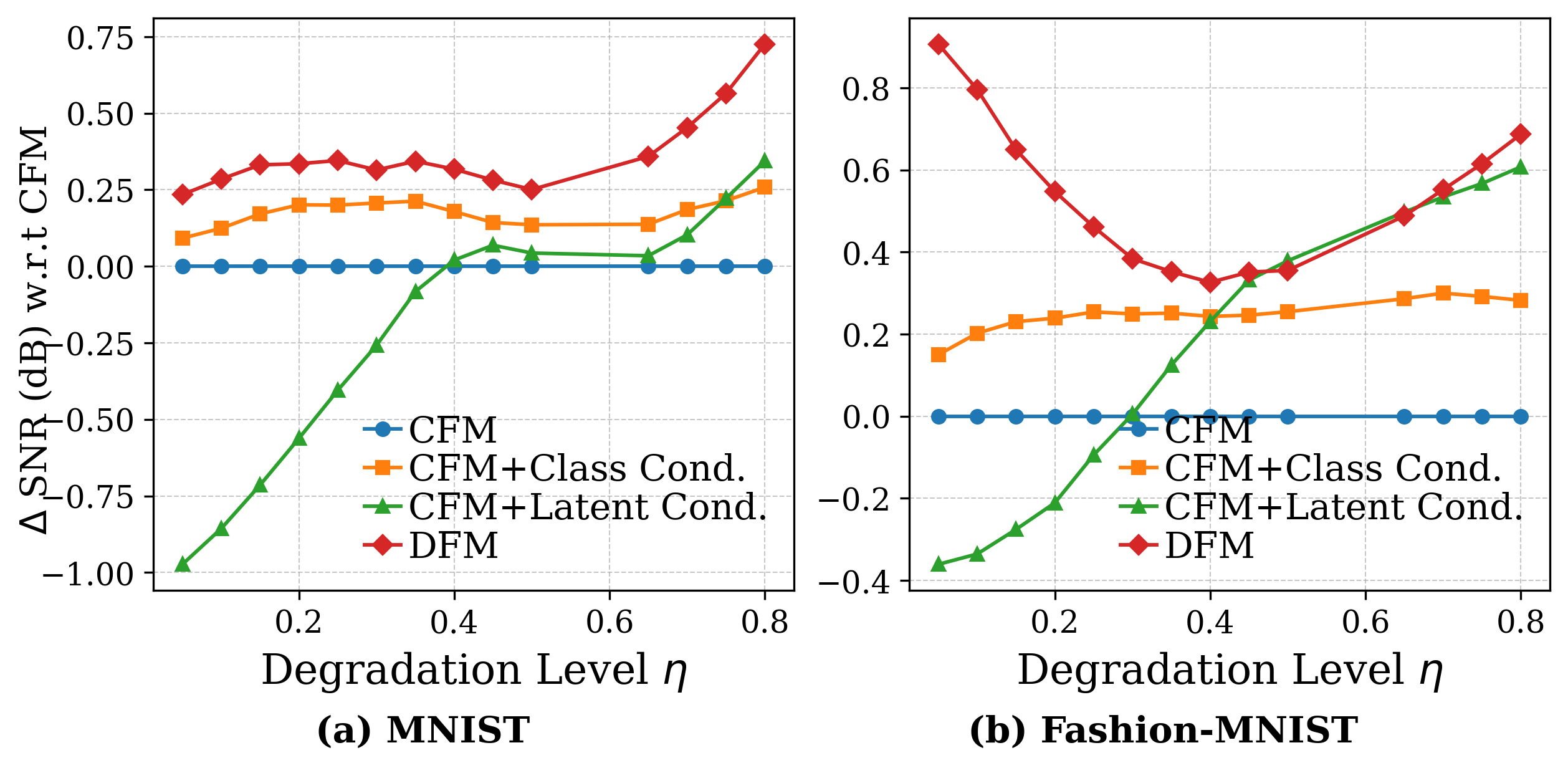}
\caption{Image denoising performance relative to \ac{CFM}.}
\label{imageperformance}
\end{figure}

\subsection{Image Denoising}
\label{Image Denoising}

To demonstrate that \ac{DFM} is not restricted to \ac{SE}, we conduct an image denoising experiment on the MNIST and Fashion-MNIST datasets using a lightweight image classifier as the discriminative encoder. We compare \ac{CFM}, \ac{CFM}+Class Cond. (where CFM additionally utilizes the ground truth label as conditioning information),  \ac{CFM}+Latent Cond. (which utilizes \ac{DFSR} as additional conditioning), and \ac{DFM}. Fig.~\ref{imageperformance} shows that \ac{DFM} consistently outperforms the baselines on both MNIST and Fashion-MNIST across all degradation levels. This highlights that \ac{DFSR} provide superior conditioning compared to semantic labels or latent features combined with explicit time conditioning. Furthermore, \ac{DFM} also supports adaptive inference for this task, allocating computational budget relative to the degradation level while preserving restoration quality. Network architectures and implementation details are provided in the Appendix~A1.14.

\section{Conclusion}

In this work, we introduce \ac{DFM}, a novel \ac{CFM} framework that replaces explicit time conditioning with discriminative Flow-State representations. Grounded in the Discriminative Flow-State Hypothesis, \ac{DFM} addresses the limitations of standard \ac{CFM} for restoration tasks by leveraging discriminative representations to dynamically adapt the generative process. Experiments across \ac{SE} and image denoising show that \ac{DFM} consistently improves restoration quality while enabling adaptive inference at a reduced computational cost, establishing discriminative conditioning as a powerful alternative to explicit time conditioning for \ac{CFM}.

\bibliography{aaai2027}

\clearpage
\appendix
\setcounter{secnumdepth}{2}
\renewcommand{\thesection}{A\arabic{section}}
\renewcommand{\thetable}{A\arabic{table}}
\renewcommand{\thefigure}{A\arabic{figure}}
\onecolumn 

\section{Appendix}

\subsection{Proof of Proposition 1 (Properties of Flow-State for OT-CFM)}
\label{app:flowgeomtry}
 \setcounter{proposition}{0}
\begin{proposition}[Properties of Flow-State under \ac{OT-CFM}]
For a transport trajectory governed by $\pi_{\text{OT}}$, the Flow-State $\flowstate(p_t)$ satisfies
\begin{equation}
    \flowstate(p_t) = (1-t) \Wtwo(p_0, p_1).
    \label{eq:app_theorem1}
\end{equation}
Consequently, $\flowstate(p_t)$ adheres to the following properties for all $t \in [0, 1]$:
\begin{enumerate}
    \item \textbf{Boundedness:} $0 \le \flowstate(p_t) \le \Wtwo(p_0, p_1)$.
    \item \textbf{Monotonicity:} $\frac{\text{d}}{\text{d}t} \flowstate(p_t) = - \Wtwo(p_0, p_1) < 0$, $\forall t < 1$.
\end{enumerate}
\end{proposition}
\begin{proof}
We begin with proving \eqref{eq:app_theorem1}. First note that for $t=1$, \eqref{eq:app_theorem1} trivially holds. Hence, let assume in the following that $t\neq 1$. We utilize the linear interpolation path 
\begin{equation}
    x_t - x_1 = (1-t)x_0 + t x_1 - x_1 = (1-t)(x_0 - x_1).
\end{equation}
Squaring its Euclidean norm yields
\begin{equation}
    \|x_t - x_1\|^2 = (1-t)^2 \|x_0 - x_1\|^2.
\end{equation}
First, preserving the optimal coupling $\pi_{\text{OT}}$ for $t=0$, we drop the infimum in the definition of $\Wtwo^2(p_t, p_1)$ to get the upper bound
\begin{align}
    \Wtwo^2(p_t, p_1) &\le \mathbb{E}_{(x_t, x_1) \sim \pi_t} \|x_t - x_1\|^2 \nonumber \\
    &= (1-t)^2 \mathbb{E}_{(x_0, x_1) \sim \pi_{\text{OT}}} \|x_0 - x_1\|^2 = (1-t)^2 \Wtwo^2(p_0, p_1).
\end{align}
Note that equality holds in the second step as the infimum over all couplings is attained by $\pi_{\text{OT}}$.
Second, we prove equality by contradiction. Assume there exists an alternative valid coupling $\pi_t'$ at time $t$ that achieves a strictly lower cost
\begin{equation}
    \mathbb{E}_{(x_t, x_1) \sim \pi_t'} \|x_t - x_1\|^2 < (1-t)^2 \Wtwo^2(p_0, p_1).
\end{equation}
Using the invertible linear relation $x_0 = \frac{x_t - t x_1}{1-t}$, we can construct a corresponding coupling $\pi_0'$ at $t=0$
\begin{align}
    \mathbb{E}_{(x_0, x_1) \sim \pi_0'} \|x_0 - x_1\|^2 &= \mathbb{E}_{(x_t, x_1) \sim \pi_t'} \left\| \frac{x_t - t x_1}{1-t} - x_1 \right\|^2 \nonumber \\
    &= \frac{1}{(1-t)^2} \mathbb{E}_{(x_t, x_1) \sim \pi_t'} \|x_t - x_1\|^2 \nonumber \\
    &< \frac{1}{(1-t)^2} (1-t)^2 \Wtwo^2(p_0, p_1) = \Wtwo^2(p_0, p_1).
\end{align}
This contradicts the definition of $\Wtwo^2(p_0, p_1)$ as the absolute infimum over all valid couplings. Thus, equality holds 
$$\Wtwo(p_t, p_1) = (1-t) \Wtwo(p_0, p_1)$$
and \eqref{eq:app_theorem1} is proven. Differentiating with respect to $t$ yields 

$$\frac{\text{d}}{\text{d}t} \mathcal{S}(p_t) = - \Wtwo(p_0, p_1) < 0 ,$$ satisfying Property (2).
Property (1) directly follows from the non-negativity of the Wasserstein metric and since $(1-t) \le 1$ for $t \in [0, 1]$.
\end{proof}

\subsection{Proof of Proposition 2: Ambiguity of Time Coordinate}
\label{app:proof_ambiguity}

\begin{proposition}[Ambiguity of Time $t$]
\label{propapp:Ambiguity}
Let the initial state be defined as $x_0^{(\alpha,\beta)} = \alpha x_1 + \beta \nu$, where $\alpha, \beta \in \mathbb{R}^+$ are scaling parameters, $\nu$ is an additive noise signal, and $x_1$ is the target clean signal. For the linear interpolation path $x_t^{(\alpha,\beta)} = (1-t)x_0^{(\alpha,\beta)} + t x_1$ with $t\in[0,1)$, there exist a distinct parameter set $\{\alpha', \beta', t'\}$ with $t \neq t'$, $t'\in[0,1)$ and $(\alpha, \beta) \neq (\alpha', \beta')$ with $\alpha', \beta'\in \mathbb{R}^+$ such that
\begin{equation}
p\left(x_t^{(\alpha,\beta)}\right) = p\left(x_{t'}^{(\alpha',\beta')}\right).
\end{equation}

\end{proposition}

\begin{proof}
    We show that there exists $\alpha',\beta',t'$ such that $x_t^{(\alpha,\beta)}=x_{t'}^{(\alpha',\beta')}$ from which the claim follows immediately. For $x_t^{(\alpha,\beta)}=x_{t'}^{(\alpha',\beta')}$, the coefficients of $x_1$ and $\nu$ have to match, i.e.,
    \begin{equation}
        (1-t)\beta = (1-t')\beta' \qquad \Rightarrow \qquad \beta' = \frac{1-t}{1-t'}\beta.
    \end{equation}
    Note that as $\beta' > 0$ the chosen $\beta'$ is valid and as $\beta\neq\beta'$ it follows $t\neq t'$ as required.
    For the coefficient of $x_1$, we get
    \begin{equation}
        (1-t)\alpha +t = (1-t')\alpha' + t' \qquad \Rightarrow \qquad \alpha' = \frac{1-t}{1-t'}\alpha + \frac{t-t'}{1-t'},
    \end{equation}
    which is a valid choice if $\alpha'>0$, which is satisfied if $(1-t)\alpha +t \geq t'$.
    
    As the last step we have to show that such an $\alpha'$ exists. Let's assume for contradiction that for given $t$ and $\alpha$ no $t'$ exists such that $(1-t)\alpha +t \geq t'$, i.e.,
    \begin{equation}
        (1-t)\alpha +t < t'\quad \forall t'\in[0,1).
    \end{equation}
    As the left hand side is positive this is obviously wrong and we conclude the proof.
\end{proof}

\subsection{Proof of Proposition 3: Upper Bound of Flow-State}
\label{app:upperbound flowstate}

\begin{proposition}[Upper Bound of Flow-State]
\label{propapp:upper_boundmain}
Assuming without loss of generality  that $\mathbb{E}\|x_1\|^2 = \mathbb{E}\|\nu\|^2 = 1$, we obtain
\begin{equation}
    \flowstate(p_t) \leq |1-t| (|\alpha-1| + |\beta|)
    \label{eq:app_prop3_ineq}
\end{equation}
and in the special case of a convex combination, for $\beta = 1-\alpha$, where $x_0 = \alpha x_1 + (1-\alpha)\nu$ with $\alpha \in [0, 1]$
\begin{equation}
\Wtwo(p_t, p_1) \leq 2 |1-t| |\alpha-1|.
\end{equation}
\end{proposition}

\begin{proof}
First assume that the assumption of $\mathbb{E}\|x_1\|^2 = \mathbb{E}\|\nu\|^2 = 1$ would be violated, i.e., $\mathbb{E}\Vert{}x_1\Vert{}^2 = a$ and $\mathbb{E}\Vert{}\nu\Vert{}^2 = b$, where $a, b \neq 1$. However, then we could define a simple reparameterization of the initial state $x_t = \alpha x_1 + \beta \nu$ as $x_t = \tilde{\alpha}\tilde{x}_1 + \tilde{\beta}\tilde{\nu}$, where
\begin{equation*}
\tilde{\alpha} = \alpha \sqrt{\mathbb{E}\Vert x_1\Vert^2}, \quad \tilde{\beta} = \beta \sqrt{\mathbb{E}\Vert \nu\Vert^2}, \quad \tilde{x}_1 = \frac{x_1}{\sqrt{\mathbb{E}\Vert x_1\Vert^2}}, \quad \tilde{\nu} = \frac{\nu}{\sqrt{\mathbb{E}\Vert \nu\Vert^2}}
\end{equation*}
and the assumption of $\mathbb{E}\Vert{}\tilde{x}_1\Vert{}^2 = \mathbb{E}\Vert{}\tilde{\nu}\Vert{}^2 = 1$ would hold with the new parameterization. Hence, we assume without loss of generality that $\mathbb{E}\|x_1\|^2 = \mathbb{E}\|\nu\|^2 = 1$ in the following.

By using the definition of $x_t$, we obtain
\begin{equation}
\Vert x_t - x_1\Vert = |1-t|\ \Vert(\alpha-1)x_1 + \beta \nu\Vert.
\end{equation}
By applying the definition of the Wasserstein-2 distance, we obtain an upper bound by evaluating the expectation over an arbitrary coupling (thereby dropping the infimum)
\begin{equation}
\mathcal{S}(p_t) = \Wtwo(p_t, p_1) \leq \sqrt{\mathbb{E}\Vert x_t - x_1\Vert^2} = |1-t| \sqrt{\mathbb{E}\Vert(\alpha-1)x_1 + \beta \nu\Vert^2}.
\end{equation}
Applying the triangle inequality to the term inside the square root yields
\begin{equation*}
\leq |1-t| \sqrt{\mathbb{E}\left( |\alpha-1|\ \Vert x_1\Vert + |\beta|\ \Vert \nu\Vert \right)^2}.
\end{equation*}
We expand the expectation and apply the Cauchy-Schwarz inequality $\mathbb{E}(xy)^2\leq\mathbb{E}(x^2)\mathbb{E}(y^2)$ to obtain
\begin{align*}
&= |1-t| \sqrt{|\alpha-1|^2 \mathbb{E}\Vert x_1\Vert^2 + 2|\alpha-1|\ |\beta|\ \mathbb{E}\left(\Vert x_1\Vert\ \Vert \nu\Vert\right) + |\beta|^2\ \mathbb{E}\Vert \nu\Vert^2}\\
&\leq |1-t| \sqrt{|\alpha-1|^2 \mathbb{E}\Vert x_1\Vert^2 + 2|\alpha-1|\ |\beta|\ \sqrt{\mathbb{E}\Vert x_1\Vert^2}\sqrt{\mathbb{E}\Vert \nu\Vert^2} + |\beta|^2\ \mathbb{E}\Vert \nu\Vert^2}.
\end{align*}
Substituting $\mathbb{E}\Vert{}x_1\Vert{}^2 = 1$ and $\mathbb{E}\Vert{}\nu\Vert{}^2 = 1$ yields
\begin{equation}
= |1-t| \sqrt{|\alpha-1|^2 + 2|\alpha-1||\beta| + |\beta|^2} = |1-t| (|\alpha-1| + |\beta|).
\end{equation}
For the special case of a convex combination of $x_1$ and $\nu$, i.e., by substituting $\beta = 1 - \alpha$ with $\alpha\in[0,1]$ into the upper bound \eqref{eq:app_prop3_ineq} gives
\begin{equation}
\Wtwo(p_t, p_1) \leq |1-t| (|\alpha-1| + |1-\alpha|) = 2 |1-t| |\alpha-1|,
\end{equation}
%
which concludes the proof.
\end{proof}

\subsection{Proof of Proposition 4: Approximation of Flow-State}
\label{app:discriminative boundnessproof}
\begin{proposition}[Approximation of Flow-State]
Let $\hat{x}_1\sim p_d$ be the estimate of $x_1\sim p_1$ obtained from a discriminative restoration model trained under a mean squared error objective $\mathcal{L} = \mathbb{E}[\Vert{}\hat{x}_1 - x_1\Vert{}^2]$. 
Then, we have
\begin{equation}
    |\mathcal{S}(p_t) - \hat{\mathcal{S}}(p_t)|^2 \le \mathcal{L}.
\end{equation}
\end{proposition}

\begin{proof}
We have to upper bound the squared difference of the Flow-State $\mathcal{S}(p_t) = \Wtwo(p_t, p_1)$,
with $p_1$ denoting the true underlying data distribution of clean samples $x_1$, and $p_t$ denoting the distribution of intermediate degraded states $x_t$ and its surrogate $\hat{\mathcal{S}}(p_t) = \Wtwo(p_t, p_d)$, with $p_d$ denoting the distribution of the discriminative model's estimates (where $\hat{x}_1 = \mathcal{D}_{\Phi}(\mathcal{E}_\phi(x_t)) \sim p_d$ or $\hat{x}_1 = \mathcal{F}_\Omega(x_t)$).

%

\paragraph{Step 1: Bounding the Estimation Error via Inverse Triangle Inequality.}
Using the inverse triangle inequality for the Wasserstein-2 metric
\begin{equation}
|\Wtwo(p_t, p_1) - \Wtwo(p_t, p_d)| \le \Wtwo(p_1, p_d),
\end{equation}
we establish the following inequality on the absolute difference between the true and approximate Flow-States
    
%
\begin{equation}
|\mathcal{S}(p_t) - \hat{\mathcal{S}}(p_t)| \le \Wtwo(p_1, p_d).
\label{eq:app_prop4_wIneq}
\end{equation}

\paragraph{Step 2: Upper Bounding the Wasserstein Distance with Reconstruction Loss.}
By definition, the Wasserstein-2 distance between the estimated distribution $p_d$ and the true distribution $p_1$ is the infimum over all valid joint couplings $\pi \in \Pi(p_d, p_1)$
\begin{equation} 
\Wtwo(p_d, p_1) = \inf_{\pi \in \Pi(p_d, p_1)} \sqrt{\mathbb{E}_{(\hat{x}_1, x_1) \sim \pi}[\|\hat{x}_1 - x_1\|^2]}.
\end{equation}
%
The pairing of the target $x_1$ with the estimate $\hat{x}_1 = \mathcal{F}_\Omega(x_t)$ produced from the corresponding degraded state $x_t$ is described by a joint distribution of $(\hat{x}_1, x_1)$ whose first marginal is $p_d$ and whose second marginal is $p_1$. Note that there exist different joint distributions, which are collectively denoted by the set $\Pi(p_d, p_1)$, and that the Wasserstein-2 distance is based on an infimum over all of these joint distributions. Hence, by picking a particular joint distribution $p(x_t,x_1)$, we obtain an upper bound on the Wasserstein-2 distance
\begin{equation}
    \Wtwo(p_d, p_1) \le \left( \mathbb{E}_{p(x_t, x_1)}
    \left[ \Vert \mathcal{F}_\Omega(x_t) - x_1 \Vert^2 \right] \right)^{1/2}
    = \sqrt{\mathcal{L}},
\end{equation}
where $\mathcal{L} = \mathbb{E}_{p(x_t, x_1)} [ \Vert \mathcal{F}_\Omega(x_t) - x_1
\Vert^2 ]$ denotes the supervised mean squared reconstruction error of the
discriminative model. Substituting into~\eqref{eq:app_prop4_wIneq} yields
\begin{equation}
    \left| \mathcal{S}(p_t) - \hat{\mathcal{S}}(p_t) \right| \le \sqrt{\mathcal{L}},
\end{equation}
and squaring both sides concludes the proof.
\end{proof}

\subsection{Proof of Proposition 5: Flow-State Preservation}
\label{app:Flow-State PreservationProof}

\begin{proposition}[Flow-State Preservation]
\label{propapp:flow_state_preservation}
Assuming the discriminative model's architecture forms a Markov chain $x_t \to z \to \hat{x}_1$, where $z = \mathcal{E}_{\phi}(x_t)$ and $\hat{x}_1$ is the estimated clean signal. Then,
\begin{equation}
    \mathcal{I}(z;x_1)\geq \mathcal{I}(\hat{x}_1;x_1),
\end{equation}
where $\mathcal{I}$ denotes the mutual information.
\end{proposition}

\begin{proof}
We model the discriminative restoration process as the Markov chain
\[
x_1 \rightarrow x_t \rightarrow z \rightarrow \hat{x}_1,
\]
assuming that the decoder $\mathcal{D}_{\Phi}$ receives no direct skip connection from the degraded input $x_t$. Under this assumption, the claim directly follows from the data processing inequality.


\end{proof}

\subsection{Flow-State Preservation}
\label{app:Flow-State preserve}

For the Flow-State preservation experiments discussed in
Section~\ref{res:Flow-State Preservation}, we follow the FlowSE
training framework\footnote{\url{https://github.com/seongq/flowmse}},
using the same training parameters and implementation settings as
those employed for the \ac{DFM} experiments in
Section~\ref{sec:seexp}. For predicting the Flow-State-related quantities, iSNR, gSNR and $\zeta$, we adapt the
PESQNet architecture from the SNR-Aligned DiffSE
implementation\footnote{\url{https://github.com/yh-jun/SNR-Aligned_diffSE/blob/main/sgmse-bbed/sgmse/backbones/snrnet.py}}.
The anchor-based predictor containing approximately
$1.26$M parameters processes the signal pair $\{\mu_{\zeta}, \mathcal{F}_{\Omega}(\mu_\zeta)\}$ described in
Section~\ref{res:Flow-State Preservation}. For the latent-based predictor, we replace all the
two-dimensional convolutional layers with one-dimensional
convolutions to accommodate latent inputs of shape
$\texttt{batch size} \times \texttt{channels} \times \texttt{time frames}$. This modification results in a substantially
smaller model for the latent-based predictor relative to the anchor-based predictor with approximately $447$k parameters.

The predictors are trained for $800$k steps using the same
DNS Challenge training data~\cite{reddy2020interspeech} employed for
the \ac{CFM} experiments. The evaluation is performed on the DNS
Challenge synthetic non-reverberant test set, which contains
$150$ samples. For each sample, we generate $10$ intermediate
states according to
\eqref{eq:path} by uniformly sampling one interpolation coordinate
$\zeta$ from each of ten non-overlapping intervals of width $0.1$ over $[0,1]$.
This results in a total of $1{,}500$ evaluation samples.

\subsection{Flow-State Monotonicity}
\label{sec:Latent Geometry}

In the following, we further evaluate \ac{DFSR} $z=\mathcal{E}_{\phi}(x_{\zeta})$ with respect to Flow-State monotonicity for different degradation levels expressed in terms of \ac{SNR}. For this experiment, a reference latent set is first constructed by extracting frame-level representations from clean speech using the frozen DCCRN \cite{hu2020dccrn} encoder. The same utterance is then corrupted with two different noise types (rain and machinery noise as discussed in Section \ref{res:Flow-State Preservation}) across \ac{SNR}s ranging from $-20$~dB to $60$~dB, and the resulting latent representations are compared with the clean latent set.

\begin{figure}[!t]
\centering
\includegraphics[width=\columnwidth]{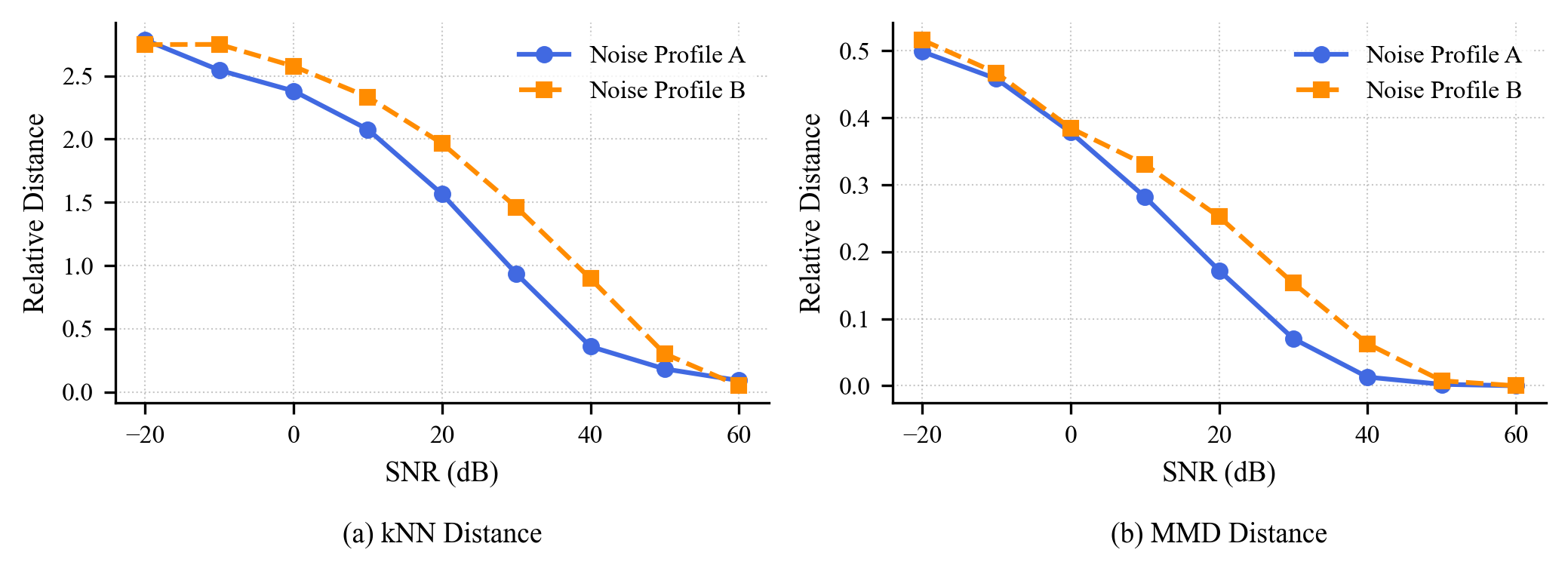}
\caption{Relative kNN and MMD distances between noisy latent representations and the clean latent set across different SNR levels. Both measures decrease monotonically with increasing SNR, indicating that intermediate samples progressively approach the clean latent representations.}
\label{fig:latent_geometry_snr}
\end{figure}

To quantify latent proximity, we employ two measures: the average $k$-nearest-neighbor (kNN) distance and the Maximum Mean Discrepancy (MMD) distance~\cite{gretton2012kernel} between noisy and clean latent representations. Both measures are reported relative to their corresponding clean-reference values to account for the  variability of the clean latent set.

\begin{table}[!t]
\centering
\begin{tabular}{lc}
\toprule
Metric & Value \\
\midrule
Pearson ($r$) & $-0.9645$ \\
Spearman ($\rho$) & $-0.9832$ \\
Monotonicity violations & $0/8$ \\
\bottomrule
\end{tabular}
\caption{Monotonic relationship between SNR and latent proximity.}
\label{tab:mmd_results}
\end{table}

Fig.~\ref{fig:latent_geometry_snr} shows that both, the relative kNN distance and the MMD distance, decrease monotonically as \ac{SNR} increases, indicating that \ac{DFSR} $z$ move progressively closer to the clean latent set as degradation decreases. Table~\ref{tab:mmd_results} further confirms this behavior, yielding strong monotonic relationships with \ac{SNR} (Spearman $\rho=-0.9832$) without a single monotonicity violation (i.e., the relative distance for both measures decreased across the  $8$ evaluated \ac{SNR} levels).

These results further support our experimental observation provided in Section~\ref{sec:monotonicity}, that a discriminative latent space serves as a meaningful descriptor of the underlying Flow-State as the relative distance toward the clean latent representations shrinks as the degradation decreases.

\subsection{Speech Enhancement}
\label{app:SE Extras}

As mentioned in \eqref{eq:dfm}, the \ac{DFM} training objective is formulated as
$$
\mathcal{L}_{\mathrm{DFM}}(\theta)
=\mathbb{E}\left[\left\|
v_\theta(x_{\zeta},z,x_0) -v^\star\right\|_2^2\right],$$
where $z \in \mathbb{R}^{C_{\text{in}} \times T_{\text{orig}}}$ is the \ac{DFSR} derived from a pretrained discriminative encoder $\mathcal{E}_{\phi}$, with $C_{\text{in}}$ denoting the number of output channels of the encoder and $T_{\text{orig}}$ the original number of outputted time frames. The representation is first projected onto a $256$-dimensional feature space using a linear layer and subsequently linearly interpolated along the temporal axis to match the input-frame dimension $T$ of the velocity network $v_\theta$. Here, $T$ denotes the number of short-time Fourier transform (STFT) frames obtained from the degraded observation $x_0$ using the STFT configuration specified in Section~\ref{sec:seexp}. This temporal alignment is necessary because the pretrained discriminative model might use different STFT parameters than those used for training the velocity network of our \ac{DFM}, resulting in different time resolutions. Formally, this feature projection (applied along the channel dimension) and temporal interpolation are defined as
\begin{align}
\bar{z} &= \mathbf{W} z, \\
c &= \text{Interp}(\bar{z}, T),
\end{align}
where $\mathbf{W} \in \mathbb{R}^{256 \times C_{\text{in}}}$ is the learnable weight matrix of the linear projection, resulting in the projected \ac{DFSR} $\bar{z} \in \mathbb{R}^{256 \times T_{\text{orig}}}$. The function $\text{Interp}(\cdot)$ denotes linear interpolation along the time axis, yielding the aligned conditioning signal $c \in \mathbb{R}^{256 \times T}$. For supporting multi-scale conditioning in the backbone network, these temporal features $c$ are additionally downsampled by a factor of $2$ at each successive residual block using a 1D finite impulse response (FIR) filter as used originally in the NCSN++~\cite{song2021score} architecture along the time dimension. We further condition each residual block of the backbone NCSN++~\cite{song2021score} network used in \ac{DFM} analogously to the additive conditioning method utilized  for time conditioning in~\cite{lee2025flowse,richter2023speech} as
\begin{equation}
h_{\text{out}} = h + \bar{c},
\end{equation}
where $h \in \mathbb{R}^{C_{\text{res}} \times F \times T}$ denotes the intermediate feature map of the residual block of the NCSN++ backbone network with $C_{\text{res}}$ channels and $F$ frequency bins. The conditioning signal $c$ is broadcasted across the frequency dimension (denoted as $\bar{c}$) before addition to $h$.

\begin{table}[!t]
\centering
\setlength{\tabcolsep}{3pt} 
\begin{tabular}{llccc}
\toprule
\textbf{Category} & \textbf{Method} & \textbf{Conditioning} & \begin{tabular}[c]{@{}c@{}}\textbf{SI-SDR}\\($\uparrow$)\end{tabular} & \begin{tabular}[c]{@{}c@{}}\textbf{PESQ}\\($\uparrow$)\end{tabular} \\
\midrule
Unprocessed & Noisy & -- & 9.06 & 1.58 \\
\midrule
Discriminative & DCCRN & NA & 17.36 & 2.91 \\
& GCRN & NA & 16.71 & 2.63 \\
& TaylorSENet & NA & 16.59 & 2.65 \\
& ULCNet & NA & 16.70 & 2.56\\
\midrule
CFM & FlowSE & $t$& 18.99 & 2.86 \\
\midrule
\multirow{4}{*}{Proposed} & \multirow{4}{*}{\textbf{DFM (Ours)}} & GCRN & \textbf{20.04} & 2.95 \\
& & TaylorSENet & 19.65 & 2.90 \\
& & ULCNet & 19.49 & 2.95 \\
& & DCCRN & 19.63 & \textbf{2.96} \\
\bottomrule
\end{tabular}
\caption{\ac{DFM} results on the DNS Challenge synthetic non-reverberant test set for utilizing discriminative features from different  discriminatively trained \ac{SE} methods.}
\label{tab:DNSE1}
\end{table}
\subsection{Experiment with Different Discriminative Models}
\label{sec:Discriminative_Models}

To evaluate the generalization capability of our  proposed \ac{DFM} for different discriminative representations, we extract the intermediate discriminative features $z$ from four diverse \ac{SE} architectures: DCCRN~\cite{hu2020dccrn}, ULCNet~\cite{shetu2023ultra}, TaylorSENet~\cite{ijcai2022p582}, and GCRN~\cite{tan2019learning}. 

\subsubsection{DCCRN} 
The DCCRN~\cite{hu2020dccrn} employs a complex-valued UNet-like symmetric encoder-decoder architecture with a recurrent bottleneck to estimate a complex ideal ratio mask. We extract the conditioning information $z \in \mathbb{R}^{1024 \times T_{\text{orig}}}$ from the output of the affine transform following the complex-valued LSTM. We implement the DCCRN model following the official implementation available at GitHub \footnote{\url{https://github.com/PoKoHA/Speech_Enhancement-DCCRN}}.

\subsubsection{ULCNet} 
ULCNet~\cite{shetu2023ultra} is an ultra-lightweight  network designed for low complexity \ac{SE}. It utilizes two stage network architecture, where the first stage enhances the signal magnitude and the second stage refines the magnitude and the phase jointly. We extract the intermediate discriminative features from the output of the first stage, $z \in \mathbb{R}^{256 \times T_{\text{orig}}}$. We implement the ULCNet model by adapting the implementation provided at GitHub\footnote{\url{https://github.com/narrietal/Fast-ULCNet}}.

\subsubsection{TaylorSENet}
TaylorSENet~\cite{ijcai2022p582} formulates \ac{SE} as a Taylor series expansion, decomposing the task into a coarse zeroth-order approximation and a higher-order refinement stage. We utilize the coarse zeroth-order estimate, reshaped across the frequency and feature dimensions, as the discriminative representation $z \in \mathbb{R}^{322 \times T_{\text{orig}}}$. We implement the TaylorSENet model following the official implementation available at GitHub\footnote{\url{https://github.com/Andong-Li-speech/TaylorSENet}}.

\subsubsection{GCRN}
The GCRN~\cite{tan2019learning} learns a complex-valued spectral mapping using an encoder-decoder network equipped with gated linear units (GLUs) to enhance information flow. The real and imaginary parts of the spectrogram are estimated by separate decoders. We extract the latent conditioning features $z \in \mathbb{R}^{1024 \times T_{\text{orig}}}$ directly from the output of the LSTM bottleneck. We implement the DCCRN model following the official implementation available at GitHub\footnote{\url{https://github.com/JupiterEthan/GCRN-complex}}.

As shown in Table~\ref{tab:DNSE1}, our proposed \ac{DFM} consistently outperforms both the time-conditioned \ac{CFM} (FlowSE) baseline and all standalone discriminative models in terms of SI-SDR and PESQ, regardless of the chosen discriminative representation. Specifically, \ac{DFM} conditioned on GCRN features achieves the highest SI-SDR ($20.04$~dB), while conditioning on DCCRN yields the highest PESQ ($2.96$). A similar trend is also observed for SCOREQ-based metrics.

\begin{table}[!t]
\centering
\setlength{\tabcolsep}{4pt}
\begin{tabular}{lcccc}
\toprule
\textbf{Input} & \textbf{STFT Match} & \begin{tabular}[c]{@{}c@{}}\textbf{SI-SDR}\\($\uparrow$)\end{tabular} & \begin{tabular}[c]{@{}c@{}}\textbf{PESQ}\\($\uparrow$)\end{tabular} \\
\midrule
$\mu_\zeta$ & No & 19.48 & 2.91 \\
$\mu_\zeta$ & Yes & 19.87 & 2.94 \\
$x_\zeta$ & No & 19.74 & 2.89 \\
$x_\zeta$ & \textbf{Yes} & \textbf{19.86} & \textbf{2.96} \\
$x_0$ & No & 19.55 & 2.92 \\
$x_0$ & Yes & 19.63 & 2.96 \\
\bottomrule
\end{tabular}
\caption{Ablation study of DFM on the DNS Challenge synthetic non-reverberant test set where the discriminative model used for conditioning is trained with different inputs and STFT configurations.}
\label{tab:DFM_ablation}
\end{table}

\subsection{Impact of Input Features for Discriminative Models on DFM Performance}
\label{sec:Impact_Discriminative_Models}

We also evaluated the impact of different input features and STFT parameters matched and mismatched scenarios for the discriminative model used in our proposed \ac{DFM}. For this purpose, we trained the DCCRN \cite{hu2020dccrn} model with three different input features ($\mu_\zeta$, $x_\zeta$, and $x_0$) using STFT parameters with matched (FFT length $510$, window length $510$, hop length $128$) and mismatched (FFT length $512$, window length $400$, hop length $100$) configurations relative to the original \ac{DFM} training.

As observed in Table~\ref{tab:DFM_ablation}, \ac{DFM} conditioned with features from any of the evaluated discriminative models trained with different input features obtains superior performance compared to CFM-related baselines. Discriminative models trained with $x_\zeta$ achieve the best SI-SDR and PESQ of $19.86$ and $2.96$, respectively. This can be intuitively understood, as using input $x_\zeta$ and $x_1$ as training targets can be viewed as an ARF training paradigm with a data-prediction objective for \ac{CFM}. It's important to note that discriminative models trained with $x_0$ as input also help in maintaining comparable results for \ac{DFM} (SI-SDR of $19.63$ and PESQ of $2.96$). 
Furthermore, our experiments suggest that training the discriminative model with STFT parameters matched to the \ac{DFM} training helps improve performance. However, the performance improvement seems to be minimal, suggesting that \ac{DFM} relies more on global discriminative representation alignment than on exact frame-level alignment. These results highlight that for \ac{DFM}, no particular training pipeline is needed for the discriminative model; rather, even an out-of-the-box discriminative model can be used to infer discriminative Flow-State representations.

\begin{table}[!t]
\centering
\setlength{\tabcolsep}{4.5pt} 
\begin{tabular}{lcccc}
\toprule
\textbf{Method} & \begin{tabular}[c]{@{}c@{}}\textbf{SI-SDR}\\($\uparrow$)\end{tabular} & \begin{tabular}[c]{@{}c@{}}\textbf{PESQ}\\($\uparrow$)\end{tabular} & \begin{tabular}[c]{@{}c@{}}\textbf{SCOREQ}\\ \textbf{Ref.} ($\downarrow$)\end{tabular} & \begin{tabular}[c]{@{}c@{}}\textbf{SCOREQ}\\ \textbf{Non-Ref.}($\uparrow$)\end{tabular} \\
\midrule
Noisy & 9.06 & 1.58 & 0.93 & 2.78 \\
Clean & -- & -- & 0 & 4.60\\
\midrule
CFM  & 6.81 & 1.31 & 0.76 & 3.2\\
ARF & 6.86 & 1.31 & 0.7 & 3.35\\
\textbf{DFM (Ours)} & \textbf{7.20} & \textbf{1.34} & \textbf{0.67} & \textbf{3.46} \\
\bottomrule
\end{tabular}
\caption{Results on the DNS Challenge non-reverberant test set without using the noisy signal $x_0$ for conditioning the velocity network.}
\label{tab:DNSNy}
\end{table}

\subsection{Impact of Conditioning the Velocity Network with the Noisy Observed Signal}
\label{sec:Impact_Noisy_Signal}

In this experiment, we study the impact of using the noisy signal, i.e., the initial observation, for conditioning the velocity networks in the \ac{CFM} and \ac{DFM} training objectives. As mentioned in \eqref{eq:dfm}, the \ac{DFM} training objective is formulated as
$$
\mathcal{L}_{\mathrm{DFM}}(\theta)
=
\mathbb{E}
\left[
\left\|
v_\theta(x_{\zeta},z,x_0)-v^\star
\right\|_2^2
\right].$$
For this experiment, we omit the noisy signal $x_0$ from conditioning of the velocity network. Then, the \ac{DFM} training objective can be reformulated as
\begin{equation}
   \tilde{\mathcal{L}}_{\mathrm{DFM}}(\theta)
=
\mathbb{E}
\left[
\left\|
v_\theta(x_{\zeta},z)-v^\star
\right\|_2^2
\right] .
\end{equation}
The results shown in Table \ref{tab:DNSNy} depict that discarding the noisy signal from conditioning the velocity network significantly degrades the performance of all \ac{CFM} variants. In reference-based objective metrics, the results appear even worse than the noisy signal itself. However, in our informal listening tests, we observe that all methods still denoise sufficiently, yet suffer from an energy mismatch in the generated signal. This can be attributed to the fact that, without accessing the original noisy signal, the velocity network cannot faithfully preserve the energy of the target signal, thereby degrading performance on reference-based objective metrics. On the contrary, we observe improvements for the enhanced signals across all methods using DNN-based intrusive and non-intrusive metrics, which aligns with our informal subjective listening. Compared to the baseline methods, \ac{DFM} achieves superior performance, further proving the effectiveness of the proposed method.

\begin{algorithm}[!t]
\caption{Adaptive Inference}
\label{alg:adaptive_inference}
\begin{algorithmic}[1]
\Require Difficulty estimate $d$, max budget $N_{\max}$, reverse interval $T_{\text{rev}}$, decay $\lambda$, Gaussian noise prefactor $\sigma$, initial noisy observation $x_0$, velocity network $v_\theta$, pre-trained discriminative encoder $\mathcal{E}_\phi$
\State Compute $w_i = \frac{\exp[-\lambda(i-1)]}{\sum_{j=1}^{N_{\max}}\exp[-\lambda(j-1)]}$ for $i=1 \dots N_{\max}$
\State Compute cumulative thresholds $\tau_i = T_{\text{rev}} \sum_{j=1}^{i} w_j$
\State Select $\widehat{N} \gets \min \{ i\in\mathbb{N} : d \le \tau_i \}$
\State Set $x^{(0)}=x_0 + \sigma\epsilon$ with $\epsilon \sim \mathcal{N}(0,1)$
\If{$\widehat{N} \le 1$}
    \State $\gamma \gets T_{\text{rev}}$
\Else
    \State Compute vector of normalized tail weights $w^{\text{tail}}$ of length $\widehat{N}-1$ using decay $\lambda$ using \eqref{eq:tailweights}
    \State Concatenate $\gamma \gets \left[\gamma_1, d\  w^{\text{tail}}\right]$
    \State Normalize $\gamma \gets \frac{T_{\text{rev}}}{\sum_{n=1}^{\widehat{N}} \gamma_n}\gamma$
\EndIf
\For{$n=0 \dots \widehat{N}-1$}
    \State $x^{(n+1)} \gets x^{(n)} +\gamma_n v_\theta\left( x^{(n)},
\mathcal{E}_{\phi}(x^{(n)}), x_0 \right)$
\EndFor
\State \Return $x^{(\widehat{N})}$
\end{algorithmic}
\end{algorithm}

\subsection{Adaptive DFM Inference}
\label{app:adaptive_inference}

To implement the adaptive inference, we investigate three Flow-State-related quantities; the \ac{DFM} interpolation coordinate $\zeta$,  iSNR, and gSNR (SNRs are normalized $\in [0,1]$ following Equation (14)), as described in Section~\ref{sec:adaptive_inference}, as the pseudo-sample-level difficulty estimate $d \in [0,1]$ (e.g., $d = 1-\hat{\zeta}$ or $d = 1-\text{iSNR}$), because these Flow-States-indicator quantities  also give some indication of how close an observed sample is to the clean speech manifold. Furthermore, inspired by the exponential reverse integration schedule proposed in \cite{wang2026rethinking}, we utilize exponentially decaying weights to define computational thresholds to estimate the required \ac{NFE} for an observed noisy sample $x_0$. 

Let $N_{\max}\in\mathbb{N}$ denote the maximum computational budget (maximum \ac{NFE}), $\lambda \in (0,1]$ the decay parameter controlling the decision boundaries, and $T_{\text{rev}} \in (0,1]$ the reverse interval representing the starting point of the reverse integration. We calculate empirically motivated normalized weights $w_i$ for each step $i = 1, \ldots, N_{\max}$ as
\begin{equation}
w_i = \frac{\exp[-\lambda(i-1)]}{\sum_{j=1}^{N_{\max}}\exp[-\lambda(j-1)]} \in (0,1].
\label{eq:tailweights}
\end{equation}
The corresponding cumulative thresholds $\tau_i$, which define the decision boundaries for a given sample's difficulty estimate $d$, are given by
\begin{equation}
\tau_i = T_{\text{rev}} \sum_{j=1}^{i} w_j, \quad i=1,\ldots,N_{\max}.
\end{equation}
The required \ac{NFE} $\widehat{N}$ is estimated as
\begin{equation}
\widehat{N} = \min \{ i \in \mathbb{N} : d \le \tau_i \}.
\end{equation}
When $\widehat{N} \ge 1$, the step-size schedule $\gamma = [\gamma_1, \gamma_2, \ldots, \gamma_{\widehat{N}}]$ is also constructed using the difficulty estimate $d$. The initial step size is set to $\gamma_1 = 1 - d$, and the subsequent step sizes $\gamma_n$ are chosen according to normalized tail weights $w^{\text{tail}}_n$ (see Algorithm~1) computed using the decay parameter $\lambda$
\begin{equation}
w^{\text{tail}}_n = \frac{\exp[-\lambda(n-1)]}{\sum_{j=1}^{\widehat{N}-1}\exp[-\lambda(j-1)]}, \quad n=1,\ldots,\widehat{N}-1.
\end{equation}
The step sizes are defined such that $\sum_{n=1}^{\widehat{N}} \gamma_n = T_{\text{rev}}$. The full procedure is summarized in Algorithm~\ref{alg:adaptive_inference}. In this work, we employ a maximum budget of $N_{\max} = 5$, $T_{\text{rev}} = 1.0$, and $\lambda = 0.3$. We empirically found these parameters based on the evaluation of
a PESQ-based reconstruction error for
$\mathcal{L}_{\mathrm{rec}}$, with the threshold set to
$\xi=0.10$, as defined in \eqref{eq:optimalN}.

\noindent\textbf{Example:} 
Consider a maximum budget $N_{\max} = 5$, $T_{\text{rev}} = 1.0$, a decay parameter $\lambda = 0.3$, and a difficulty estimate $d = 0.61$. 
First, the cumulative decision thresholds $\tau_i$ are computed from the normalized weights $w_i$, yielding $\tau = [0.3336, 0.5808, 0.7639, 0.8995, 1.0000]$. 
Evaluating the condition $d \le \tau_i$ selects an effective budget of $\widehat{N} = 3$ since $0.61 \le 0.7639$. 
Next, the step-size schedule $\gamma$ is constructed: the initial step size is set to $\gamma_1 = 1 - d = 0.3900$, and the remaining budget $d = 0.61$ is distributed across the remaining two tail steps using normalized exponential weights, resulting in $\gamma_{\text{tail}} = [0.3504, 0.2596]$. Thus, the final step-size schedule is $\gamma = [0.3900, 0.3504, 0.2596]$.

\subsection{Speech Dereverberation}
\label{app:SE DE}
To further evaluate the generalization capabilities of \ac{DFM} across various \ac{SE} tasks, we examine speech dereverberation, which relies on a convolutive signal model
\begin{equation}
x_0 = x_1 * g,
\end{equation}
where $*$ denotes the convolution operator, $g$ represents the room impulse response, and $x_1$ represents the dry speech. We adopt the same training objective for \ac{DFM} as defined in \eqref{eq:dfm}
\begin{equation*}
\mathcal{L}_{\mathrm{DFM}}(\theta) = \mathbb{E} \left[ \left\| v_{\theta}(x_{\zeta}, z,x_0) - v^\star \right\|_2^2 \right],
\end{equation*}
where $z=\mathcal{E}_{\phi}(x_{\zeta})
$ is derived from a pretrained encoder $\mathcal{E}_{\phi}$ of a discriminative speech dereverberation model.

\subsubsection{Training Details}
We use the clean speech from the DNS Challenge dataset. We simulated the reverberated signals using the \texttt{pyroomacoustics}\footnote{\url{https://github.com/LCAV/pyroomacoustics}} Python package with $T_{60}$ ranging from $0.3$ to $1\text{ s}$, following a publicly available simulation script\footnote{\url{https://github.com/sp-uhh/deep-non-linear-filter/blob/main/src/data/data_gen_var_pos.py}} with default room configurations for single-microphone setups. We simulated a total of $100$ hours of training data and an $8$-hour disjoint validation dataset. We first train a discriminative DCCRN model using the same configuration as discussed above in Section~\ref{app:SE Extras} to estimate the dry speech $x_1$. In the next step, we train the \ac{DFM} model with the same configurations as stated above for $600\text{k}$ steps.

\subsubsection{Results}
We evaluate the  proposed \ac{DFM} against the baseline DCCRN and CFM models. For evaluation, we use the DNS Challenge reverberant synthetic test set \cite{reddy2020interspeech} and the clean reference from the corresponding DNS Challenge non-reverberant synthetic test set \cite{reddy2020interspeech} as the target signal. The results shown in Table \ref{tab:DNSDR} depict that, for the dereverberation task, \ac{DFM} also outperforms or achieves comparable results to the baseline CFM method. This demonstrates the generalizability of the proposed \ac{DFM} method to other \ac{SE} tasks and non-additive distortions. 

It is important to note that, in this experiment, the performance gains—at least with respect to the SCOREQ metric; are statistically negligible. These results can be intuitively explained by the inferior performance of the discriminative DCCRN model, and can be further related to Proposition 4 regarding the approximation of the Flow-State. The inferior performance of the discriminative model causes \ac{DFSR} to diverge from the true Flow-State, hence, resulting in only marginal gains for \ac{DFM}. Furthermore, estimating dry speech with a mask-based method like DCCRN is inherently difficult; therefore, we hypothesize that with a better-defined training objective  or  a better discriminative model without masking, performance can be improved, allowing the true benefits of \ac{DFM} to be realized for this speech dereverberation task.

\begin{table}[!t]
\centering
\setlength{\tabcolsep}{4.5pt} 
\begin{tabular}{lcccc}
\toprule
\textbf{Method} & \begin{tabular}[c]{@{}c@{}}\textbf{SI-SDR}\\($\uparrow$)\end{tabular} & \begin{tabular}[c]{@{}c@{}}\textbf{PESQ}\\($\uparrow$)\end{tabular} & \begin{tabular}[c]{@{}c@{}}\textbf{SCOREQ}\\ \textbf{Ref.} ($\downarrow$)\end{tabular} & \begin{tabular}[c]{@{}c@{}}\textbf{SCOREQ}\\ \textbf{Non-Ref.}($\uparrow$)\end{tabular} \\
\midrule
Noisy &-31.55  &1.32  &0.84 & 2.99 \\
Clean &-  &  &0 &4.60  \\

\midrule
DCCRN  & -29.06 & 1.91 & 0.87 & 2.73\\
CFM & -27.64 & 2.35 & 0.59& 3.45 \\
\textbf{DFM (Ours)} & \textbf{-27.19} & \textbf{2.44} & \textbf{0.58} & \textbf{3.47} \\
\bottomrule
\end{tabular}
\caption{Results on the DNS Challenge reverberant test set for a speech dereverberation task.}
\label{tab:DNSDR}
\end{table}

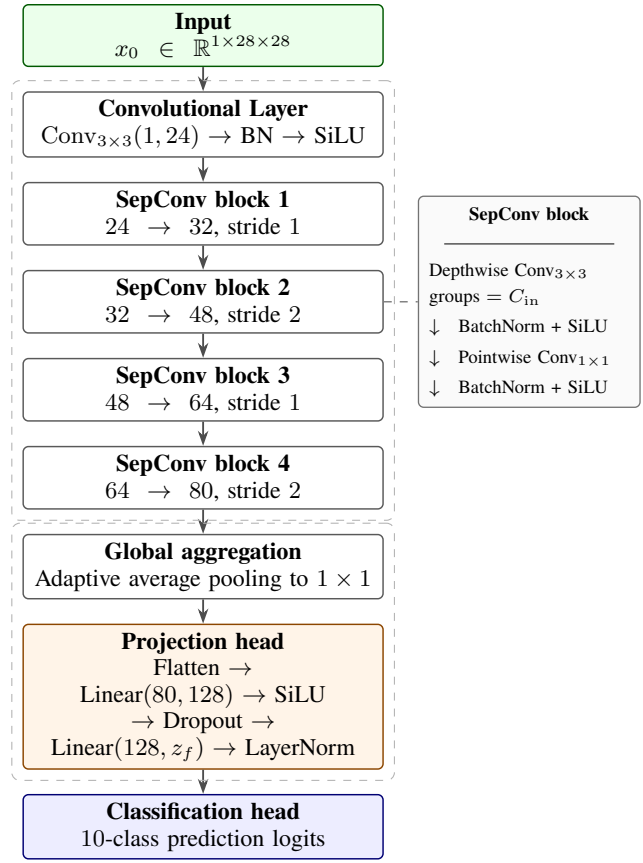
\begin{figure}[!t]
\centering
\begin{tikzpicture}[
    font=\footnotesize,
    node distance=3.2mm,
    mainblock/.style={
        draw=black!65,
        line width=0.55pt,
        rounded corners=2pt,
        fill=white,
        align=center,
        text width=4.55cm,
        minimum height=6.5mm,
        inner sep=3pt
    },
    inputblock/.style={
        mainblock,
        fill=green!8,
        draw=green!35!black
    },
    featureblock/.style={
        mainblock,
        fill=orange!9,
        draw=orange!55!black
    },
    outputblock/.style={
        mainblock,
        fill=blue!7,
        draw=blue!45!black
    },
    detailblock/.style={
        draw=black!55,
        line width=0.5pt,
        rounded corners=2pt,
        fill=gray!4,
        align=left,
        text width=2.65cm,
        inner sep=4pt,
        font=\scriptsize
    },
    stagebox/.style={
        draw=black!30,
        dashed,
        rounded corners=3pt,
        inner sep=4pt
    },
    flowarrow/.style={
        -{Stealth[length=2.0mm,width=1.4mm]},
        line width=0.7pt,
        draw=black!70
    },
    detailarrow/.style={
        dashed,
        line width=0.55pt,
        draw=black!50
    }
]

\node[inputblock] (input)
    {\textbf{Input}\\
     $x_0\in\mathbb{R}^{1\times 28\times 28}$};

\node[mainblock, below=of input] (stem)
    {\textbf{Convolutional Layer}\\
     $\operatorname{Conv}_{3\times3}(1,24)$
     $\rightarrow$ BN $\rightarrow$ SiLU};

\node[mainblock, below=of stem] (b1)
    {\textbf{SepConv block 1}\\
     $24\rightarrow32$, stride $1$};

\node[mainblock, below=of b1] (b2)
    {\textbf{SepConv block 2}\\
     $32\rightarrow48$, stride $2$};

\node[mainblock, below=of b2] (b3)
    {\textbf{SepConv block 3}\\
     $48\rightarrow64$, stride $1$};

\node[mainblock, below=of b3] (b4)
    {\textbf{SepConv block 4}\\
     $64\rightarrow80$, stride $2$};

\node[mainblock, below=of b4] (pool)
    {\textbf{Global aggregation}\\
     Adaptive average pooling to $1\times1$};

\node[featureblock, below=of pool] (projection)
    {\textbf{Projection head}\\
     Flatten $\rightarrow$ Linear$(80,128)$ $\rightarrow$ SiLU\\
     $\rightarrow$ Dropout $\rightarrow$
     Linear$(128,z_f)$ $\rightarrow$ LayerNorm};

\node[outputblock, below=of projection] (output)
    {\textbf{Classification head}\\
     $10$-class prediction logits};

\draw[flowarrow] (input) -- (stem);
\draw[flowarrow] (stem) -- (b1);
\draw[flowarrow] (b1) -- (b2);
\draw[flowarrow] (b2) -- (b3);
\draw[flowarrow] (b3) -- (b4);
\draw[flowarrow] (b4) -- (pool);
\draw[flowarrow] (pool) -- (projection);
\draw[flowarrow] (projection) -- (output);

\node[
    detailblock,
    right=4.5mm of b2.east,
    anchor=west
] (sepdetail) {
    \centering
    \textbf{SepConv block}\\[-0.1em]
    \rule{2.25cm}{0.35pt}
    \par\raggedright
    \vspace{0.25em}
    Depthwise Conv$_{3\times3}$\\
    groups $=C_{\mathrm{in}}$
    \vspace{0.2em}

    $\downarrow$\quad BatchNorm + SiLU
    \vspace{0.2em}

    $\downarrow$\quad Pointwise Conv$_{1\times1}$
    \vspace{0.2em}

    $\downarrow$\quad BatchNorm + SiLU
};

\draw[detailarrow]
    (b2.east) -- (sepdetail.west);

\node[
    stagebox,
    fit=(stem)(b1)(b2)(b3)(b4),
    label={[font=\scriptsize\bfseries, text=black!60, anchor=north west]north west:}
] (encoderbox) {};

\node[
    stagebox,
    fit=(pool)(projection),
    label={[font=\scriptsize\bfseries, text=black!60, anchor=north west]north west:}
] (featurebox) {};

\end{tikzpicture}

\caption{Architecture of the MNIST and Fashion-MNIST Image Classifier.}
\label{fig:noisy_mnist_classifier}
\end{figure}

\subsection{Image Denoising }
\label{app:Image Denoising}

In our Image denoising experiment, we utilize an interpolation-based signal model,
parameterizing \eqref{eq:additive} with $ \alpha = 1 - \eta$ and $\beta = \eta$, resulting in
\begin{equation}
x_0 = (1-\eta)x_1 + \eta \nu,
\label{eq:interpolation}
\end{equation}
where $\nu \sim \mathcal{N}(0, I)$ denotes standard Gaussian noise and  $\eta \in [0, 1]$ is a scaling  parameter.

For this experiment, we utilize the MNIST \cite{deng2012mnist} and Fashion-MNIST \cite{xiao2017fashion} datasets. We first train a discriminative classifier model to predict the image label and subsequently train a model with the proposed \ac{DFM} and \ac{CFM}-based baselines.

\subsubsection{Image Classifier Architecture:} The model employs a lightweight convolutional encoder that maps an input degraded image $x_0 \in \mathbb{R}^{1 \times 28 \times 28}$ to a latent representation $z_f \in \mathbb{R}^{e_f}$ and outputs corresponding $10$-class prediction logits. As illustrated in Fig.~\ref{fig:noisy_mnist_classifier}, the encoder begins with a convolutional layer consisting of a $3 \times 3$ convolution that increase the input channels from $1$ to $24$, followed by batch normalization and a SiLU activation. The resulting representation is processed by four depthwise-separable convolutional blocks with channel dimensions $24 \to 32 \to 48 \to 64 \to 80$. Alternating strides of $1$ and $2$ progressively reduce the spatial resolution while increasing the channel capacity. Each block comprises a $3 \times 3$ depthwise convolution, batch normalization, and SiLU activation, followed by a $1 \times 1$ pointwise convolution, batch normalization, and SiLU activation. After global adaptive average pooling, the obtained features are flattened and passed through a projection head consisting of a linear layer from $80$ to $128$ dimensions, a SiLU activation, dropout, a final linear projection to the latent representation $z \in \mathbb{R}^{e_f}$ dimensions, and a separate linear classification head maps  $z_f$ to the $10$-class prediction logits. The model has in total $41.5$k parameters.

\subsubsection{Classifier Training}We train the image classifier for $16$ epochs using a batch size of $128$, a learning rate of $10^{-3}$, a weight decay of $10^{-4}$, and the Adam optimizer. The models are trained independently on the MNIST and Fashion-MNIST datasets.

\subsubsection{Results}The classifier achieves an accuracy of $75.7\%$ on the MNIST test set and $68.9\%$ on the Fashion-MNIST test set, evaluated on the generated degraded samples using the signal model described in \eqref{eq:interpolation}. We further analyze the properties of the discriminative latent features $z_f$ on the MNIST test set. As observed in Figure~\ref{LatentImage} (a), the latent representation exhibits clear semantic structure corresponding to class labels. We also observe another prominent cluster representing heavily degraded samples, where the model struggles to extract sufficient structure. This intuitively suggests that the discriminative latent representation preserves the Flow-State, smoothly converging toward the semantic structure from the noisy observation as degradation decreases. Similarly, Figure~\ref{LatentImage}(b) illustrates that the MMD distance between the latent features and their corresponding class-wise clean features decreases as the degradation scaling factor $\eta$ is reduced, resulting in the discriminative latent features converging toward the clean feature manifold. This experiment validates our discriminative Flow-State hypothesis using a classification model for the image denoising task.

\begin{figure}[!t]
\centering
\includegraphics[width=0.9\columnwidth]{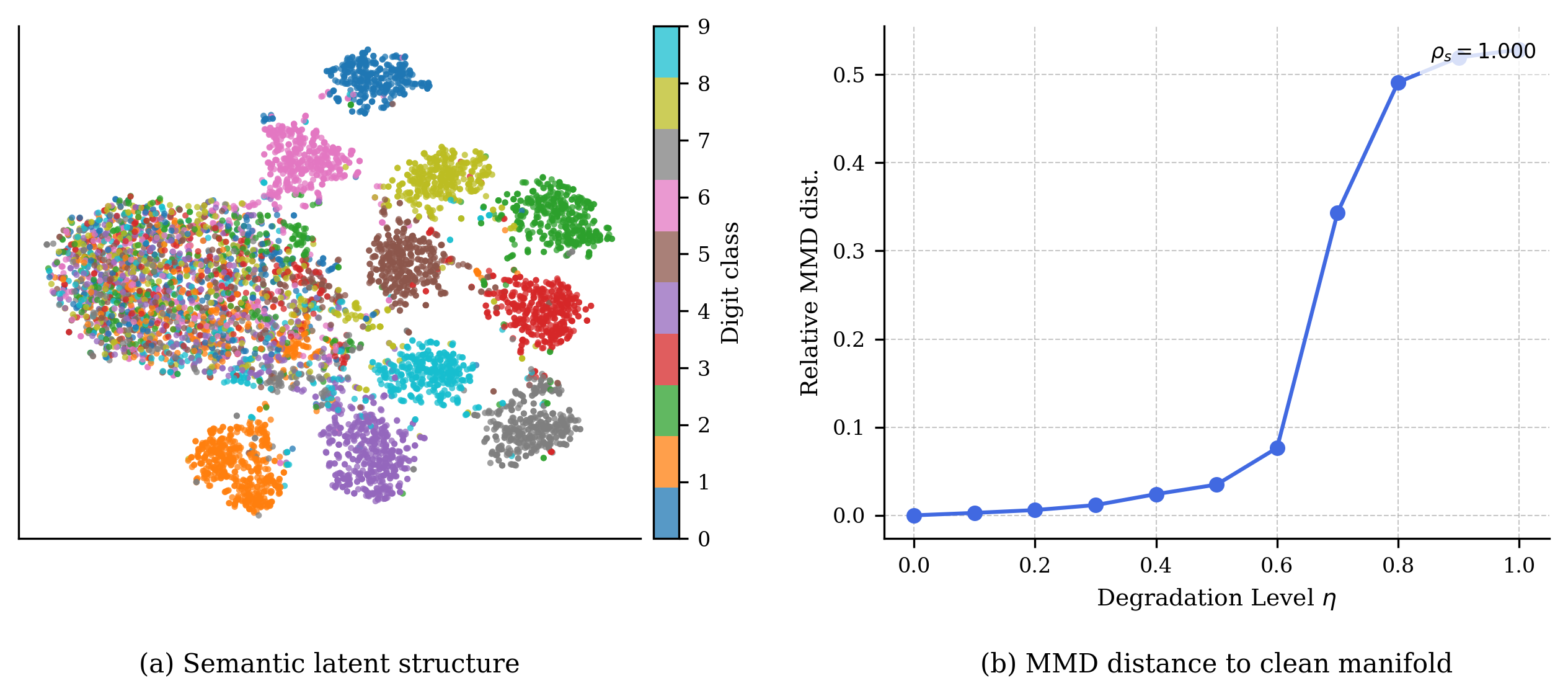}
\caption{Semantic latent structure and latent restoration geometry.}
\label{LatentImage}
\end{figure}

\subsubsection{Discriminative Flow Matching for Image Denoising}

As described in Section \ref{sec:method}, the conditional probability path is defined by
$$\mu_{\zeta}
=\zeta x_1+(1-\zeta)x_0,\qquad
x_{\zeta}= \mu_{\zeta} + \sigma_\zeta\epsilon
$$
and the \ac{DFM} training objective is formulated as
$$
\mathcal{L}_{\mathrm{DFM}}(\theta)
=\mathbb{E}\left[\left\|
v_\theta(x_{\zeta},z,x_0)-v^\star \right\|_2^2\right],$$
where $x_0$ is the initial observation of a degraded image and  $z=\mathcal{C}_{\psi}(x_{\zeta}),
$ is the Discriminative Flow-State Representation derived from the image classifier $\mathcal{C}_{\psi}$, parameterized by $\psi$.

\subsubsection{Velocity Network Architecture}
The velocity network takes the intermediate state $x_\zeta$ and degraded image $x_0$, with $[x_\zeta, x_0] \in \mathbb{R}^{2 \times 28 \times 28}$ as input, and approximates the velocity field $v^*$. The network first applies a $3 \times 3$ input convolution that projects the two-channel input to a hidden-dimensional feature map $h$. In parallel, a conditioning feature $c$, which may denote time $t$, class labels, or \ac{DFSR} $z$, is transformed by a linear projection as shown in Fig.~\ref{fig:dfm_architecture}. The resulting embedding is  added to the intermediate feature representation $h$. The conditioned features are then processed by two residual convolutional blocks, each consisting of alternating SiLU activations and $3 \times 3$ convolutions with an identity skip connection. Finally, a SiLU activation followed by a $3 \times 3$ output convolution projects the hidden representation to the  vector field $v^*$. For keeping the parameter counts comparable, for the \ac{DFM} model we use a feature dimension $64$ for hidden feature space $h$, instead of a $80$ dimensional feature space used for other baselines, namely \ac{CFM}, \ac{CFM}+Class cond. and  \ac{CFM}+Latent Cond. This results in a total of
$232$k  parameters including the overhead parameters of the image classifier for \ac{DFM} compared to $255$-k parameters for the baseline \ac{CFM}.

\begin{figure}[!t]
\centering
\begin{tikzpicture}[
    font=\footnotesize,
    node distance=2.8mm,
    block/.style={
        draw=black!65,
        line width=0.55pt,
        rounded corners=2pt,
        fill=white,
        align=center,
        text width=4.0cm,
        minimum height=6.5mm,
        inner sep=3pt
    },
    inputblock/.style={
        block,
        fill=green!8,
        draw=green!40!black
    },
    condblock/.style={
        block,
        fill=orange!8,
        draw=orange!60!black,
        text width=2.8cm
    },
    outputblock/.style={
        block,
        fill=blue!8,
        draw=blue!50!black
    },
    injectblock/.style={
        block,
        fill=purple!8,
        draw=purple!50!black
    },
    stagebox/.style={
        draw=black!30,
        dashed,
        rounded corners=3pt,
        inner sep=3.5pt
    },
    flowarrow/.style={
        -{Stealth[length=2.0mm,width=1.4mm]},
        line width=0.7pt,
        draw=black!75
    },
    condarrow/.style={
        -{Stealth[length=1.6mm,width=1.2mm]},
        dashed,
        line width=0.55pt,
        draw=purple!60!black
    }
]

\node[inputblock] (input)
    {\textbf{Inputs ($x_\zeta, x_0$)}  $\in \mathbb{R}^{2 \times 28 \times 28}$};

\node[block, below=of input] (inconv)
    {\textbf{Convlutional Layer}\\
     $\operatorname{Conv}_{3\times3}(2, \text{hidden})$};

\node[injectblock, below=of inconv] (inject)
    {\textbf{Conditioning}\\
     $h \leftarrow h + \text{MLP}(c)$};

\node[block, below=of inject] (block1)
    {\textbf{Residual Block 1}\\
     SiLU $\rightarrow$ Conv$_{3\times3}$ $\rightarrow$ SiLU $\rightarrow$ Conv$_{3\times3}$};

\node[block, below=of block1] (block2)
    {\textbf{Residual Block 2}\\
     SiLU $\rightarrow$ Conv$_{3\times3}$ $\rightarrow$ SiLU $\rightarrow$ Conv$_{3\times3}$};

\node[outputblock, below=of block2] (outconv)
    {\textbf{Output Projection}\\
     SiLU $\rightarrow$ $\operatorname{Conv}_{3\times3}(\text{hidden}, 1)$\\
     velocity $v^*$};

\draw[flowarrow] (input) -- (inconv);
\draw[flowarrow] (inconv) -- (inject);
\draw[flowarrow] (inject) -- (block1);
\draw[flowarrow] (block1) -- (block2);
\draw[flowarrow] (block2) -- (outconv);

\node[condblock, right=7mm of inject.east, anchor=west] (condproj)
    {\textbf{Conditioning MLP}\\
     $\text{Linear} \rightarrow \text{SiLU} \rightarrow \text{Linear}$};

\node[condblock, above=of condproj] (condin)
    {\textbf{Conditioning Input ($c$)}\\
     Time $t$, Label, or Disc. Features $z$};

\draw[flowarrow] (condin) -- (condproj);
\draw[condarrow] (condproj.west) -- (inject.east);

\node[
    stagebox,
    fit=(inconv)(inject)(block1)(block2),
    label={[font=\scriptsize\bfseries, text=black!60, anchor=south west]north west:}
] (mainbox) {};

\node[
    stagebox,
    fit=(condin)(condproj),
    label={[font=\scriptsize\bfseries, text=black!60, anchor=south west]north west:}
] (condbox) {};

\end{tikzpicture}

\caption{Architecture of the velocity network for image denoising task.}
\label{fig:dfm_architecture}
\end{figure}
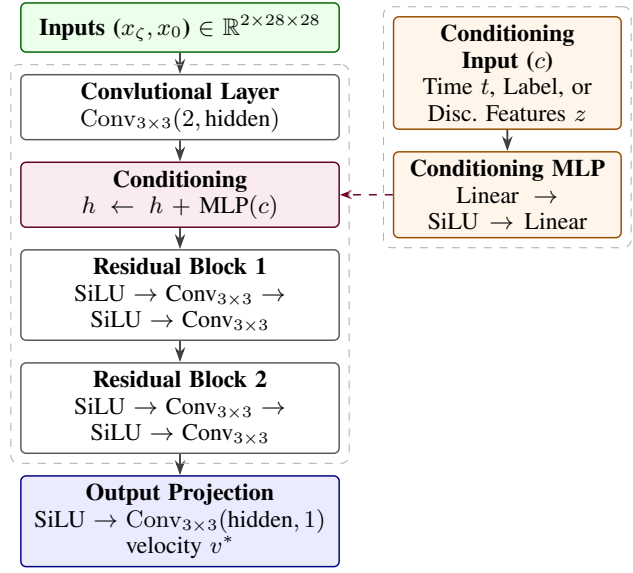

\subsubsection{Training Details} We train \ac{DFM} and all baseline models for $16$ epochs using a batch size of $128$, a learning rate of $2 \times 10^{-4}$, a weight decay of $10^{-4}$, and the Adam optimizer. The models are trained independently on the MNIST and Fashion-MNIST datasets.

\subsubsection{Results} As shown in Section~\ref{Image Denoising}, \ac{DFM} consistently outperforms the baseline methods across all degradation levels for both the MNIST and Fashion-MNIST datasets. We further evaluate adaptive inference for the image denoising task. To this end, we train a small \ac{DFM} interpolation coordinate $\zeta$ predictor comprising only three linear layers. We employ the adaptive inference algorithm discussed in Appendix~\ref{app:adaptive_inference} using a uniform Euler solver. As illustrated in Figure~\ref{AdaptiveImage}, adaptive inference performs effectively for image denoising: the algorithm assigns a higher number of function evaluations (\ac{NFE}) to heavily degraded samples and a lower \ac{NFE} to less degraded samples. Overall, it matches the performance of the fixed budget scheduler ($\ac{NFE} = 50$) while significantly reducing the average computational budget.

\begin{figure}[!t]
\centering
\includegraphics[width=0.9\columnwidth]{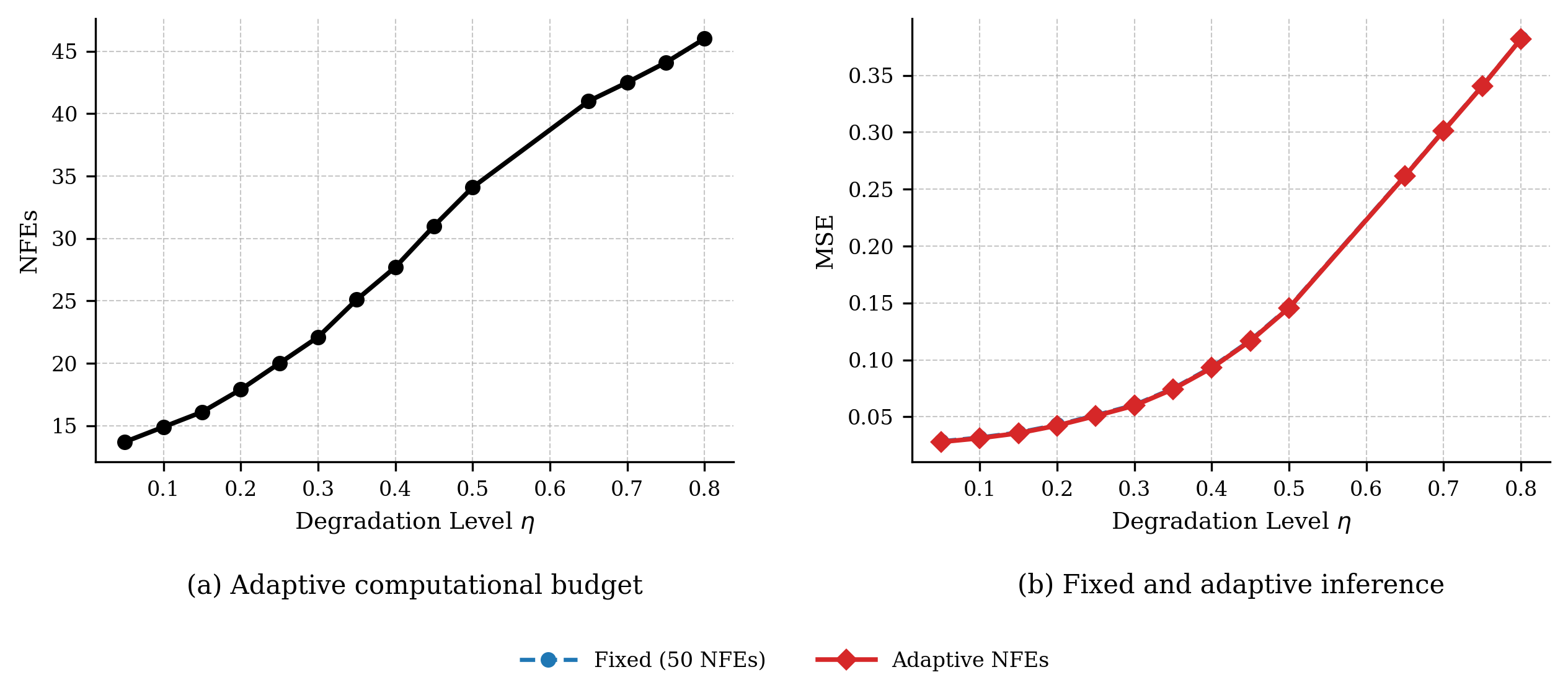}
\caption{Adaptive inference with uniform Euler solver for image denoising on the MNIST dataset.}
\label{AdaptiveImage}
\end{figure}

\subsection{Statistical Analysis of the Results}
\label{sec:statisticalanalyis}

\begin{figure}[!t]
\centering
\includegraphics[width=\columnwidth]{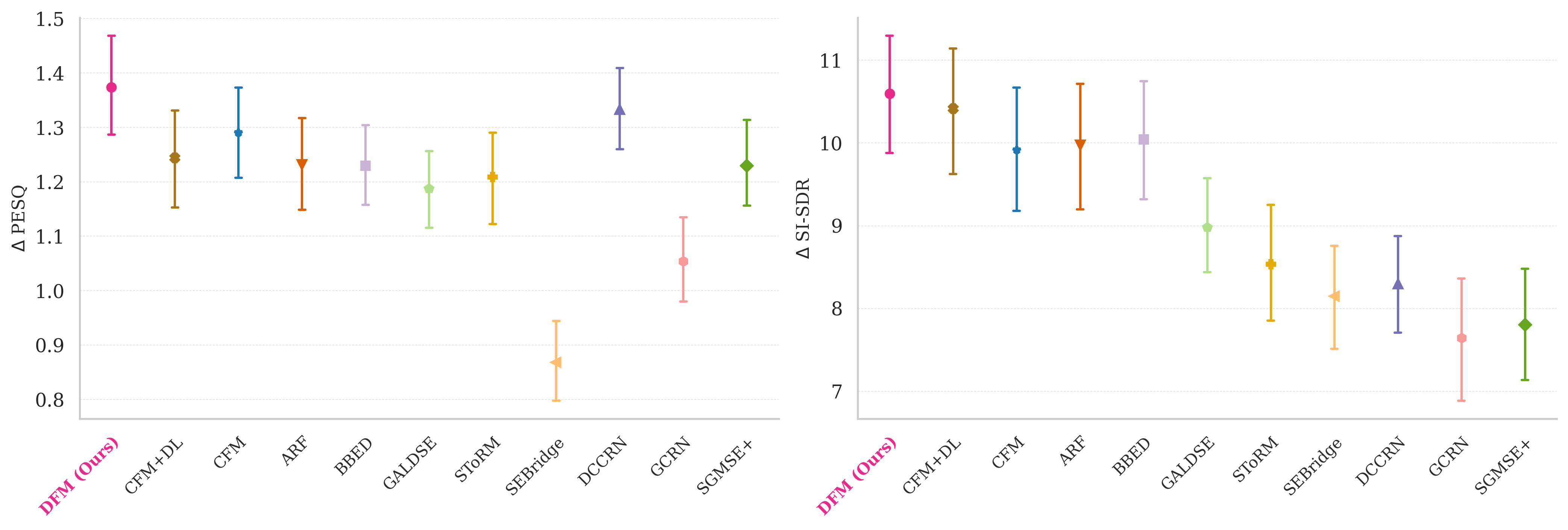}
\caption{PESQ and SI-SDR improvements along with their respective $95\%$ confidence intervals for the compared methods presented in Table~\ref{tab:DNSnoreverb1}.}
\label{fig:Boxplot}
\end{figure}

\begin{table*}[t]
\centering
\begin{tabular}{l c c c c}
\toprule
\textbf{Baseline Model} & \multicolumn{2}{c}{\textbf{PESQ}} & \multicolumn{2}{c}{\textbf{SI-SDR (dB)}} \\
\cmidrule(lr){2-3} \cmidrule(lr){4-5}
& \textbf{Baseline Mean} & \textbf{$p$-value} & \textbf{Baseline Mean} & \textbf{$p$-value} \\
\midrule
DCCRN & $2.91$ & $0.300$ & $17.36$ & $< 0.001$ \\

GCRN & $2.63$ & $< 0.001$ & $16.71$ & $< 0.001$ \\
BBED & $2.81$ & $< 0.001$ & $19.10$ & $< 0.001$ \\

SGMSE+ & $2.81$ & $< 0.001$ & $16.86$ & $< 0.001$ \\
SToRM & $2.70$ & $< 0.001$ & $17.56$ & $< 0.001$ \\
GALDSE & $2.77$ & $< 0.001$ & $18.04$ & $< 0.001$ \\
SEBridge & $2.45$ & $< 0.001$ & $17.21$ & $< 0.001$ \\

CFM & $2.86$ & $< 0.001$ & $18.99$ & $< 0.001$ \\
CFM+DL & $2.87$ & $< 0.001$ & $19.47$ & $< 0.001$ \\

ARF & $2.82$ & $< 0.001$ & $19.04$ & $< 0.001$ \\

\bottomrule
\end{tabular}
\caption{Statistical significance test results (paired Wilcoxon signed-rank test) comparing our proposed \ac{DFM} method (PESQ: $2.96$, SI-SDR: $19.63\text{ dB}$) against various baseline models.}
\label{tab:statistical_summary}
\end{table*}

We further validate the performance gains of the proposed \ac{DFM} method presented in Table \ref{tab:DNSnoreverb1} , with paired  Wilcoxon signed-rank tests against all baseline models. The statistical analysis reveals, in terms of SI-SDR, \ac{DFM}  achieves consistent improvements ($p$-value<0.001) over every competing model. In terms of PESQ, \ac{DFM}   likewise demonstrates statistically significant superiority ($p$-value<0.001) against all baselines except DCCRN. However, as stated earlier, DCCRN performs much inferiorly in terms of SI-SDR. Overall, the statistical evaluation confirms that \ac{DFM} 's performance enhancements are consistent, non-random, and highly robust across evaluation metrics. In Fig. \ref{fig:Boxplot}, we show the PESQ and SI-SDR improvements of the competing methods on the DNS Challenge non-reverberant test set presented in Table \ref{tab:DNSnoreverb1}, along with their corresponding confidence intervals, illustrating the consistent and robust performance of \ac{DFM}.

\end{document}